\documentclass[10pt]{article}
\usepackage[utf8]{inputenc}

\usepackage{custom-macros}

\title{\textbf{Parallel Derandomization for Coloring}}

\author{\textbf{Sam Coy}\thanks{E-mail: S.Coy@warwick.ac.uk. Department of Computer Science, University of Warwick, UK.
Research supported in part by the Centre for Discrete Mathematics and its Applications (DIMAP),
by an EPSRC studentship, and
by the Simons Foundation Award No. 663281 granted to the Institute of Mathematics of the Polish Academy of Sciences for the years 2021--2023.}
\\
University of Warwick
\and
\textbf{Artur Czumaj}\thanks{E-mail: A.Czumaj@warwick.ac.uk. Department of Computer Science and Centre for Discrete Mathematics and its Applications (DIMAP), University of Warwick, UK. Research supported in part by the Centre for Discrete Mathematics and its Applications, by EPSRC award EP/V01305X/1, by a Weizmann-UK Making Connections Grant, by an IBM Award, and by the Simons Foundation Award No. 663281 granted to the Institute of Mathematics of the Polish Academy of Sciences for the years 2021--2023.}
\\
University of Warwick
\and
\textbf{Peter Davies-Peck}\thanks{\red{TODO}}
\\
Durham University
\and
\textbf{Gopinath Mishra}\thanks{E-mail: Gopinath@nus.edu.sg. Department of Computer Science, National University of Singapore. The work was done when Gopinath Mishra was a Post Doctoral Fellow at the University of Warwick. The research was supported in part by the Centre for Discrete Mathematics and its Applications (DIMAP), by EPSRC award EP/V01305X/1, and by the Simons Foundation Award No. 663281 granted to the Institute of Mathematics of the Polish Academy of Sciences for the years 2021--2023.}
\\
National University of Singapore
}

\author{
\begin{tabular}[t]{c c c c}
\textbf{Sam Coy}\thanks{E-mail: S.Coy@warwick.ac.uk. Department of Computer Science, University of Warwick, UK.
Research supported in part by the Centre for Discrete Mathematics and its Applications (DIMAP),
by an EPSRC studentship, and
by the Simons Foundation Award No 663281 granted to the Institute of Mathematics of the Polish Academy of Sciences for the years 2021-2023.} &
\textbf{Artur Czumaj}\thanks{E-mail: A.Czumaj@warwick.ac.uk. Department of Computer Science and Centre for Discrete Mathematics and its Applications (DIMAP), University of Warwick, UK. Research supported in part by the Centre for Discrete Mathematics and its Applications (DIMAP), by EPSRC award EP/V01305X/1, by a Weizmann-UK Making Connections Grant, by an IBM Award, and by the Simons Foundation Award No. 663281 granted to the Institute of Mathematics of the Polish Academy of Sciences for the years 2021--2023.} &
\textbf{Peter Davies-Peck}\thanks{E-mail: Peter.W.Davies@durham.ac.uk. Department of Computer Science, Durham University, UK.} &
\textbf{Gopinath Mishra}\thanks{E-mail: Gopinath@nus.edu.sg. Department of Computer Science, National University of Singapore. The work was done when Gopinath Mishra was a Post Doctoral Fellow at the University of Warwick. The research was supported in part by the Centre for Discrete Mathematics and its Applications (DIMAP), by EPSRC award EP/V01305X/1, and by the Simons Foundation Award No. 663281 granted to the Institute of Mathematics of the Polish Academy of Sciences for the years 2021--2023.} \\
\small University of Warwick &
\small University of Warwick &
\small Durham University &
\small National University of Singapore
\end{tabular}
}
\date{}

\begin{document}

\maketitle

\begin{abstract}
Graph coloring problems are among the most fundamental problems in parallel and distributed computing, and have been studied extensively in both settings. In this context, designing efficient \emph{deterministic} algorithms for these problems has been found particularly challenging.

In this work we consider this challenge, and design a novel framework for derandomizing algorithms for coloring-type problems in the \emph{Massively Parallel Computation} (MPC) model with sublinear space. We give an application of this framework by showing that a recent $(degree+1)$-list coloring algorithm by Halldorsson et al. (STOC'22) in the LOCAL model of distributed computation can be translated to the MPC model and efficiently derandomized. Our algorithm runs in $O(\log \log \log n)$ rounds, which matches the complexity of the state of the art algorithm for the $(\Delta + 1)$-coloring problem.
\end{abstract}

\newpage
\section{Introduction}
\label{sec:intro}

The \emph{Massively Parallel Computation} (\MPC) model, introduced over a decade ago by Karloff \etal \cite{KSV10}, is \comments{a contemporary standard theoretical model} for parallel algorithms. \comments{The model has evolved through the successful emulation of parallel and distributed frameworks, including but not limited to MapReduce \cite{mapreduce}, Hadoop \cite{hadoop}, Dryad \cite{dryad}, and Spark \cite{spark}. Drawing inspiration from classical models of parallel computation (e.g., PRAM) and distributed models (e.g., \CONGESTEDC), MPC incorporates numerous shared attributes with these established paradigms}. We focus on the \emph{sublinear local space} MPC regime, where machines have local space $\lspace = O(n^{\phi})$ for any arbitrary \comments{fixed} constant $\phi \in (0,1)$, where $n$ is the number of nodes in the graph. \comments{Note that $\phi$ is a parameter of the model}. This model has attracted a lot of attention recently, see, e.g., \cite{ANOY14,ASSWZ18,BKM20,BKS17,BBDFHKU19,BHH19,CFGUZ19,CC22,CDPsparse,CDP21,CDP20,CLMMOS18,GGKMR18,GKU19,GU19,GSZ11,LMOS20}. 
Recent works have provided algorithms for fundamental graph problems such as connectivity, approximate matching, maximal matching, maximal independent set, and $(\Delta+1)$ coloring.

While the main focus on \MPC algorithms has been typically concentrating on the design of randomized algorithms, an area which has seen a recent surge in interest in the \MPC model is the design of \emph{deterministic algorithms} which are as efficient as their randomized counterparts. In related models of distributed computation this task is \comments{sometimes} impossible: there are lower bounds in the \LOCAL model which separate the randomized complexity from the deterministic complexity of several fundamental algorithmic problems (see, e.g., \cite{CKP19}). However, a combination of the global co-ordination and of the computational power of individual machines within the \MPC model allows for a broader selection of tools for derandomizing \MPC algorithms. For example, \comments{a particularly effective tool in this context} is the \emph{method of conditional expectations}, which has been recently used to derandomize several algorithms for fundamental problems in the \MPC model\cite{BKM20,CC22,CDP21,FGG22}.

In this paper we develop a general derandomization framework, providing a useful tool for translating some class of randomized \LOCAL algorithms to deterministic \MPC in a black-box manner. Previous works have either been heavily tailored to specific algorithms \cite{BKM20,CC22,CDP21,FGG22} or have had severe limitations on graph degree \cite{CDPcompstab}. Our approach allows generic derandomization of a large class of algorithms over a much larger degree range.

As an application of our derandomization framework, we consider the task of designing fast deterministic \MPC algorithms for graph coloring --- one of the most fundamental algorithmic primitives, extensively studied in various settings for several decades. Graph coloring problems have been playing \comments{a prominent} role in distributed and parallel computing, not only because of their numerous applications, but also since some variants of coloring problems naturally model typical symmetry breaking problems, as frequently encountered in decentralized systems (see, e.g., \cite{BE13} for an overview of early advances).
Parallel graph coloring has been studied since the 1980s \cite{BE13,Karloff85}, and nowadays $(\Delta+1)$-coloring and $(2\Delta-1)$-edge-coloring\footnote{\emph{$(\Delta+1)$-coloring} problem is to color a graph of maximum degree $\Delta$ using $\Delta+1$ colors; \emph{$(2\Delta-1)$-edge-coloring} problem is to color the edges of a graph of maximum degree $\Delta$ using $2\Delta-1$ colors.} are considered among the most fundamental graph problems --- benchmark problems in the area (throughout the paper, $\Delta$ refers to the maximum degree of the input graph).

In particular, we study the parallel complexity of a natural generalization of the $(\Delta+1)$-coloring problem, the problem of \textbf{\emph{(degree+1)-list coloring} (\DILC)}. In the \DILC problem, for a given graph $G = (V,E)$, each node has an input palette of acceptable colors of size one more than its degree, and the objective is to find a proper coloring using these palettes. The recent increasing interest in \DILC has been largely caused by the generality and applicability of \DILC, since, for example, given a partial solution to a $(\Delta+1)$-coloring problem, the remaining coloring problem on the uncolored nodes is an instance of \DILC. It also appears as a subproblem in more constrained coloring problems, e.g., as a subroutine in distributed $\Delta$-coloring algorithms (see, e.g., \cite{FHM23}) and in edge-coloring algorithms (see, e.g., \cite{Kuhn20}). 

\hide{
While it is easy to design a simple linear-time (sequential) greedy algorithm for \DILC, the parallel and distributed complexity of \DILC is less well understood. Clearly, the problem is not easier than the $(\Delta+1)$-coloring problem and its variant the $(\Delta + 1)$-list coloring problem
\footnote{In $(\Delta + 1)$-list coloring each node has a palette of $\Delta+1$ many colors for its disposal (rather than $\deg+1$ many, as in~\DILC).}, but the challenge of dealing with nodes having color palettes of greatly different sizes seems to make the problem significantly more difficult. Still, it has been observed that by using techniques developed in \cite{FHK16,Kuhn20}, one can deterministically reduce \DILC to $(\Delta + 1)$-list coloring with only an $O(\log\Delta)$ multiplicative and $O(\log^*n)$ additive overhead in the number of rounds. However the logarithmic complexity gap is still significant and until very recently, this gap has been elusive for the most efficient distributed and parallel algorithms for \DILC. The first advance (in the distributed setting) has come only very recently, when 
Halld{\'o}rsson, Kuhn, Nolin, and Tonoyan \cite{hknt_local_d1lc} presented a randomized $O(\log^3\log n)$-rounds distributed algorithm for \DILC in the \LOCAL distributed model. 
Following this work, Halld{\'o}rsson, Nolin, and Tonoyan \cite{HNT22} extended the framework and showed that \DILC can be solved (w.h.p.) in $O(\log^5\log n)$-rounds in the distributed \CONGEST model. 
In the most recent advance, Coy \etal \cite{CCDM23} obtained a deterministic constant-time rounds algorithm for \DILC in the \CONGESTEDC distributed model.
All these bounds for \DILC match the state-of-the-art complexity for the simpler $(\Delta+1)$-coloring problem in \LOCAL \cite{CLP20}, \CONGEST \cite{HKMT21}, and \CONGESTEDC \cite{CDP20} models, respectively.}


\subsection{Our contribution}
\label{subsec:our-contribution}

In this paper, we develop a framework for generic derandomization of \LOCAL algorithms in the deterministic low-space \MPC model. This framework uses pseudorandom generators and the method of conditional expectations, which are known techniques, but significantly extends prior work in generality and applicability.

Our main derandomization theorem is \Cref{lem:fullderand}. Since this theorem involves some technical definitions, we leave the statement to \Cref{sec:normalderand}. However, informally, the result shows that any randomized \LOCAL algorithm that can be decomposed into short subprocedures with some natural properties can be efficiently derandomized in \MPC.

As an application of our framework, we show that \DILC can be solved efficiently deterministically in the \emph{Massively Parallel Computation} (\MPC) model with sublinear local space, matching the complexity of the state-of-the-art \MPC algorithms for the simpler $(\Delta+1)$-coloring and $(\Delta+1)$-list coloring problems. We first show how to combine the \DILC framework for the \LOCAL model due to Halld{\'o}rsson \etal \cite{hknt_local_d1lc} with the techniques developed in earlier works on the \MPC model, to obtain a randomized \MPC algorithm for \DILC working in $O(\log\log\log n)$ rounds, w.h.p. Then we apply our novel derandomization framework to derandomize this algorithm.

\hide{
\paragraph*{MPC model}
The \emph{Massively Parallel Computation} (\MPC) model, introduced over a decade ago by Karloff \etal \cite{KSV10}, is a nowadays standard theoretical model for parallel algorithms. The model has been developed on the basis of its successful modeling of parallel and distributed frameworks such as MapReduce \cite{mapreduce}, Hadoop \cite{hadoop}, Dryad \cite{dryad}, and Spark \cite{spark}, and it shares many similarities with classical models of parallel computation (e.g., PRAM) and distributed models (e.g., \CONGESTEDC). In this paper, we focus on the \emph{sublinear local spacer} MPC regime, in which machines have local space $\lspace = O(n^{\phi})$ for any arbitrary constant $\phi \in (0,1)$, where $n$ is the number of nodes in the graph. This model has attracted a lot of attention recently, see, e.g., \cite{ANOY14,ASSWZ18,BKS17,BHH19,BKM20,BBDFHKU19,CC22,CFGUZ19,CDP20,CDP21,CDPsparse,CLMMOS18,GGKMR18,GKU19,GU19,GSZ11,LMOS20}. Recent works have provided many algorithms for fundamental graph problems such as connectivity, approximate matching, maximal matching, maximal independent set, and $(\Delta+1)$ coloring.

\sam{Insert here: introduction to derandomization in \MPC?}

\orange{An area which has seen a recent surge in interest in the \MPC model is the design of deterministic algorithms which are as efficient as their randomized counterparts. In related models of distributed computation this is often impossible: there are lower bounds in the \LOCAL model which separate the randomized complexity from the deterministic complexity of several fundamental algorithmic problems\dots
The global co-ordination possible within the \MPC model provides some tools for derandomizing \MPC algorithms. A parallel implementation of the use of the \emph{method of conditional expectations} has been used to derandomize several algorithms for fundamental problems in the \MPC model \cite{CC22,FGG22}\samaside{I'm aware there are several other results which should be cited here\dots}.}

\sam{Insert text about the use of PRGs in MPC derandomization}

\sam{Introduce framework, give result}
}

\begin{theorem}[\textbf{\DILC}]
\label{thm:main:deter}
Let $\spacexp \in (0,1)$ be an arbitrary constant. There exists a deterministic algorithm that, for every $n$-node graph $G=(V,E)$, solves the \DILC problem using $O(\log\log\log n)$ rounds, in the sublinear local space \MPC model with local space $\lspace = O(n^{\spacexp})$ and global space $O(m+n^{1+\spacexp})$.
\end{theorem}

Observe that the bound in \Cref{thm:main:deter} matches the state-of-the-art bound for the complexity of the simpler $(\Delta+1)$-coloring problem in the (sublinear local space) \MPC model (see \cite{CFGUZ19} for the randomized bound and \cite{CDP21} for the deterministic bound). Furthermore, the recently developed framework connecting the complexity of \LOCAL and sublinear local space \MPC algorithms (see \cite{CDPcompstab,GKU19}), provides some evidence that our upper bound cannot be asymptotically improved, unless the complexity of the $(\Delta+1)$-coloring problem is 
$(\log\log n)^{o(1)}$
in the \LOCAL model. This is because \cite{CDPcompstab,GKU19} show that for a class of component stable algorithms and conditioned on the so-called 1-vs-2-cycles conjecture, no sublinear local space \MPC algorithm can run faster than the logarithm of the complexity of \LOCAL algorithms. (Still, even conditioned on the 1-vs-2-cycles conjecture, it might be conceivable that a non-component stable randomized \MPC sublinear local space algorithm can solve $(\Delta+1)$-coloring in $o(\log\log\log n)$ rounds, and further, we do not have any good enough \LOCAL lower bounds for coloring, and so maybe an 
$(\log\log n)^{o(1)}$-rounds
\LOCAL algorithm is possible.)
%
Finally, notice that Roughgarden \etal \cite{RVW18} showed that proving any super-constant lower bound in the sublinear local space \MPC for any problem in \model{P} would separate $\model{NC}^1$ from \model{P}, making any \emph{unconditional} super-constant (strongly sublinear local space) \MPC lower bound unlikely.

\subsection{Technical overview}
\label{subsubsec:technical-contribution}


A natural approach to design efficient \MPC algorithms is to simulate a \LOCAL algorithm 
using the so-called graph exponentiation approach (see, e.g., \cite{ASSWZ18,Ghaffari16,GU19,LW10}): in the first round each node learns its 2-hop neighborhood, then its 4-hop neighborhood, and \comments{so on}. In some cases we can even obtain deterministic algorithms in low-space \MPC which are exponentially faster than the corresponding \emph{randomized} complexity in \LOCAL~\cite{CDPcompstab}. However, doing so is often not easy, and faces two main challenges: the first is how to efficiently derandomize algorithms, and the second is that large neighborhoods may not fit onto machines for high-degree instances, which renders challenging many common subroutines in \LOCAL algorithms (e.g., in the coloring setting, computing an almost-clique decomposition).

We address the first problem, of derandomization, by developing a general derandomization framework for transferring some class of randomized \LOCAL algorithms to deterministic low-space \MPC ones. Our framework relies on the use of pseudorandom generators (PRGs) combined with the method of conditional expectations. The approach of using PRGs in derandomization has been used in the past (see, e.g., \cite{CDPcompstab,CDP21}), but faces a major difficulty in general application. This difficulty is that the stringent constraint of the local space requirement of sublinear local space \MPC causes PRGs to fail on a non-trivial proportion of nodes, even if the base randomized procedure succeeds with high probability. Furthermore, attempts to defer or change the outputs of these failed nodes can cause a chain reaction of nodes failing to meet the success requirements of the randomized procedure. To derandomize $(\Delta+1)$-coloring, \cite{CDP21} addressed this issue by painstakingly analyzing the effects of PRG failures throughout the course of the (highly complex) base randomized algorithm, but this was highly \comments{tailored} to the specific algorithm and not easily generalizable. \cite{CDPcompstab} did provide a form of general framework, but this required collecting the neighborhoods needed for entire \LOCAL algorithms onto machines, and therefore only worked for low values of $\Delta$. These low-$\Delta$ examples sufficed for \cite{CDPcompstab} since the aim was to show that component-stable lower bounds could be surpassed in some cases, but did not aid derandomization of general~algorithms.

We overcome these difficulties by providing a suitable framework that formalizes (see \Cref{prng-derandomizable}) a collection of properties that allow us to fully derandomize randomized algorithm using the PRGs. The main technical difference from \cite{CDPcompstab} is that rather than derandomizing full algorithms in one go, we instead decompose them into short subprocedures with certain useful properties. We then have the space necessary to derandomize these subprocedures on individual machines, allowing us to tackle much \comments{higher-degree instances}.

The framework and its analysis summarized in \Cref{lem:fullderand} form the main technical contribution of our work. We hope that our framework, and in particular, the key lemma (\Cref{lem:fullderand}), will prove useful as a powerful black-box derandomization technique in \MPC.

We apply this framework to \DILC as our main application. However, for instances with very high degree, we must still reduce the degree before we can fit even $1$-hop neighborhoods onto machines. To do so, we employ a \emph{deterministic recursive sparsification} approach similar to \cite{CDP21,CDP20}, where we repeatedly partition an instance of \DILC with maximum degree $\Delta$ into $n^\delta$ \DILC instances, each with maximum degree $\Delta / n^\delta$. Here $\delta$ with $0 \le \delta \le \phi$ is to be fixed later, and $\phi$ is our local space parameter, i.e., $\lspace = O(n^\phi)$. All but one of these instances are \emph{valid} \DILC instances and so can be solved \comments{using this recursive sparsification if the degree is still too high}, and the final instance can only be solved when it is determined which colors are unused in the other instances. In this way, we can reduce the maximum degree of the \DILC instances which we have to solve to an arbitrarily small polynomial in $n$, which is small enough to apply the derandomization framework.

\subsection{Related work}
\label{sec:related-work}

\subsubsection*{Parallel derandomization.}
Our work relies on some sparsification and derandomization techniques developed for parallel and distributed coloring algorithms. Derandomization techniques such as the method of conditional expectations are long-studied and well-understood in classical sequential computation, and were first employed in linear-space \MPC (actually, in the mostly equivalent \CONGESTEDC model) by \cite{CPS20}. This was followed by a work extending the techniques to low-space \MPC \cite{CDPsparse}. These works both applied the method of conditional expectations to families of bounded-independence hash functions, a technique for reducing the length of random seeds required for randomized algorithms. However, they are only applicable to algorithms that do not heavily exploit the independence of random choices.

Many algorithms (including all known sublogarithmic \LOCAL coloring algorithms) do not appear to have this property, and their analyses effectively use $\Delta$-wise independence or higher, which is too high to efficiently apply families of bounded-independence hash functions. So, some subsequent works on parallel derandomization instead employ pseudorandom generators (PRGs) to reduce the seed space. PRGs are again a well-studied topic in classical computing (see, e.g., \cite{Vadhan12} for an introduction), and do not have the restriction on independence. The drawback is that PRGs achieving optimal parameters are only known existentially, and computing them requires exponential-time computation --- but this is generally permitted in \MPC, and even if not, the PRG can be pre-computed sequentially and hard-coded onto \MPC machines. \cite{CDP21} used PRG-based derandomization to give an $O(\log\log\log n)$-round deterministic low-space \MPC algorithm for $\Delta+1$-coloring, and \cite{CDPcompstab} used it for general derandomization, in order to demonstrate that the method of conditional expectations could be exploited to surpass the component-stable lower bounds of \cite{GKU19}. However, as mentioned, the derandomization results therein were only for low-degree instances.

Implementations of the method of conditional expectations have also recently been used for derandomization in related distributed models, 
see e.g., \cite{BKM20, GHK18, Kuhn20}. 

\subsubsection*{Parallel and distributed coloring.}
Our application to \DILC continues a long line of research studying the parallel and distributed \comments{complexity} of graph coloring problems. For the references to earlier work on distributed coloring algorithms we refer to the monograph by Barenboim and Elkin \cite{BE13} (see also the influential papers by Linial \cite{Linial87,Linial92}). We will discuss here only more recent advances (and final results) for the four most relevant coloring problems, $(\Delta+1)$-coloring, $(\Delta+1)$-list-coloring, \DILC, and $\Delta$-coloring. \comments{After extensive research in the field of distributed computing concerning the $(\Delta+1)$-coloring problem, we understand its complexity for the \LOCAL, \CONGESTEDC, and also for the \MPC model}. For \CONGESTEDC (and also for \MPC with linear memory, $\lspace = O(n)$), we \comments{now know} how to solve $(\Delta+1)$-coloring in a constant number of rounds, see \cite{chang2018optimal,CDP20}. For the \LOCAL model, after a very long line of research, the current state of the art upper bound for randomized algorithm is $\widetilde{O}(\log^2\log n)$ \cite{chang2018optimal, RG20}.~\footnote{$\widetilde{O}(f)$ hides a  polynomial term in $\log f$.}~\footnote{\comments{Through out this paper, $\log^k x$ denotes $(\log x)^k$.}} The best know deterministic algorithm is of $\widetilde{O}(\log ^2 n)$ rounds\cite{GG23} in the \LOCAL model.

For the sublinear local space \MPC, it is known that the $(\Delta+1)$-coloring algorithm due to Chang et al. \cite{CFGUZ19} can be combined with the network decomposition result of \cite{RG20} to obtain a randomized $O(\log\log\log n)$-round \MPC algorithm, which is currently the state-of-the-art for the problem; this result was derandomized by Czumaj \etal \cite{CDP21}. Furthermore, all algorithms mentioned above for $(\Delta+1)$-coloring can be extended to solve also $(\Delta+1)$-list-coloring.

For the \DILC problem, which is a generalization of $(\Delta+1)$-coloring and $(\Delta+1)$-list-coloring, there have not been many comparable bounds until the very recent work of Halld{\'o}rsson \etal \cite{hknt_local_d1lc}.  In particular, \DILC  admits a randomized $\widetilde{O}(\log^2\log n)$-round distributed algorithm in the \LOCAL model \cite{hknt_local_d1lc,GG23} matching the state-of-the-art complexity for the $(\Delta+1)$-coloring problem \cite{chang2018optimal,GG23}. In \cite{hknt_local_d1lc}, the authors significantly extended the \comments{approaches} for $(\Delta+1)$-coloring (in particular, \comments{to allow efficient management of nodes} of various degrees).
As a byproduct, the framework of Halld{\'o}rsson \etal \cite{hknt_local_d1lc} can be incorporated into a constant-round \MPC algorithm assuming the local \MPC space is slightly \emph{superlinear}, i.e., $O(n \log^4n)$ \cite[Corollary 2]{hknt_local_d1lc}.
We make extensive use of the framework laid out by Halld{\'o}rsson \etal \cite{hknt_local_d1lc} in their algorithm for \LOCAL in the design of our \DILC algorithm.
A similar approach has been applied recently for the \CONGEST model in \cite{HNT22, GK21}, solving \DILC in $\widetilde{O}(\log^3\log n)$ \CONGEST rounds, w.h.p. Very recently, a deterministic constant-round algorithm for \DILC has been obtained for the \CONGESTEDC model \cite{CCDM23}, settling the complexity of \DILC for \CONGESTEDC.



\hide{
\Artur{This and next paragraph are to be skipped.}In comparison to the perhaps more famous coloring problems listed above, relatively little has been written about the problem of $\Delta$-coloring in the parallel and distributed setting. In the \LOCAL model, the problem was first discussed in a paper by Panconesi and Srinivasan \cite{local_nature_delta_coloring}, and that algorithm was recently revisited and improved by Ghaffari \etal \cite{DBLP:journals/dc/GhaffariHKM21}.
In distributed models with communication bandwidth restrictions (and so, the setting closer to the model we consider here), the first breakthrough came only recently, in a paper by Fischer, Halld{\'o}rsson, and Maus \cite{FHM23}. They gave a \CONGEST algorithm which gives a $\Delta$-coloring of the graph in $O(\poly \log \log n)$ rounds if $\Delta = \omega(\log^3 n)$, which decreases to $O(\log^* n)$ rounds if $\Delta = \omega(\log^{21} n)$. Their result for large values of $\Delta$ was a substantial improvement over the state-of-the-art in any distributed model ($O(\log n)$ rounds previously being the best known), and combined with the result given by \cite{DBLP:journals/dc/GhaffariHKM21} for sub-polylogarithmic values of $\Delta$, they also achieved a $O(\poly\log\log n)$ round algorithm for $\Delta$-coloring in the \LOCAL model.

We are not aware of any direct prior work on $\Delta$-coloring in any model of parallel computing. An $O(\log n)$ round algorithm for \PRAM follows straightforwardly from the $O(\log n)$-round MIS algorithm of Luby \cite{Luby86}: one can find an MIS and defer coloring those nodes, the remaining nodes all gain at least $1$ slack, and therefore form a $(\Delta+1)$-coloring instance. This algorithm can be transformed into a $O(\log n)$ round \MPC algorithm, using well-known results relating to the simulation of \PRAM algorithms in \MPC. We note that the complexity of our algorithm represents a double-exponential improvement over this.}

\hide{\subsection*{Notations}

\begin{itemize}
\item $G(V,E)$-- Graph with $n$ nodes.

\item $d_{G}(v)$--- The degree of $v$ in $G$.

\item $N_{G}(v)$--- The neighborhood of $v$ in $G$.

\item $G[W]$--- The subgraph of $G$ induced by $W \subseteq V$.

\item $O(n^\spacexp)$--- The space bound of each machine \peter{I changed this because the parts from \cite{CDP20} already use $\delta$ for something else, and we'll need to compare $\delta$ and $\eps$ later.}
\end{itemize}
} 

\hide{
\subsection{Paper overview}
\label{sec:overview}

\Anote{To be updated.}%
In \Cref{sec:preliminaries} we give some notation and preliminaries on distributed and parallel coloring.
In \Cref{sec:randomized_d1lc_intuition} we \Artur{TBD}
In \Cref{sec:derandomization-framework} we present our framework for efficient derandomization of coloring problems, and give an example application.
In \Cref{sec:deterministic_d1lc} we show how our framework can be used to transform the randomized \LOCAL algorithm for \DILC given by Halld\'orsson \etal \cite{hknt_local_d1lc} into a deterministic \MPC algorithm which runs in exponentially fewer rounds, within a certain degree range.
In \Cref{sec:d1lc_for_all_degree_ranges} we show that we can extend this algorithm to a deterministic \DILC algorithm for all degree ranges.
\Artur{We should mention that some arguments and proofs are deferred to Appendix.}
}

\section{Preliminaries}
\label{sec:preliminaries}

\label{sec:setting}



\subsection{Notations, the problem, and the model}
For $k \in \mathbb{N}$, $[k]$ denotes the set $\{1,\ldots,k\}$. For $a,b \in \mathbb{N}$, $[a,b]$ denotes the set of integers in $\{a,a+1,\ldots,b\}$. We consider a graph
$G=(V,E)$ with $V$ as the node set and $E$ as the edge set with $|V|=n$ and $|E|=m$.
The set of neighbors of a node $v$ is denoted by $N(v)$ and the degree of a node $v$ is denoted by $d(v)$. For a node $v$, $\Psi(v)$ denotes the list of colors in the color palette of node $v$ and $p(v)$ denotes the size of $\Psi(v)$. The maximum degree of any node in $G$ is denoted by $\Delta$. As we go on coloring the nodes of the graph $G$, the graph will change and the color palettes of the nodes will also change. Often, we denote the current (rather than the input) graph by $G$. For all graphs we consider, we have $p(v)\geq d(v)+1$. For a subset $X \subseteq V$, $G[X]$ denotes the subgraph induced by $X$ and $m(X)$ denotes the number of edges in $G[X]$.

\subsubsection*{Degree+1 list coloring (D1LC).}
%

The \emph{degree+1 list coloring (D1LC) problem} is for a given graph $G = (V,E)$ and given color palettes $\Psi(u)$ assigned to each node $u \in V$, such that $\size{\Psi(u)} \geq d(u)+1$, the objective to find a proper coloring of nodes in $G$ such that each node is assigned to a color from its color palette and no edge in $G$ is monochromatic.

\subsubsection*{Massively Parallel Computation model.}
We consider the \emph{Massively Parallel Computation} (\MPC) model, which is a parallel system with some number of machines, each of them having some local space~\lspace. 
At the beginning of computation, each machine receives some part of the input, with the constraint that it must fit within its local space. In our case, for the \DILC problem, the input is a set of $n$ nodes, $m$ edges, and $n$ color palettes of total size $O(n+m)$. Hence we will require that the number of machines is $\Omega(\tfrac{n+m}{\lspace})$, for otherwise the input would not fit the system. The computation on an \MPC proceeds in synchronous rounds. In each round, each machine processes its local data and performs an arbitrary local computation on its data without communicating with other machines. At the end of each round, machines can exchange messages. Each message is sent only to a single machine specified by the machine that is sending the message. All messages sent and received by each machine in each round have to fit into the machine's local space. Hence, the total number of messages sent by any machine and received by any machine is bounded by \lspace, and the total amount of communication across the whole \MPC is bounded by \lspace times the number of machines. At the beginning of the next round, each machine can process all messages received in the previous round. When the algorithm terminates, machines collectively output the solution.


Observe that if a single machine can store the entire input, then any problem (like, e.g., \DILC) can be solved in a single round, since no communication is required. In order for our algorithms to be as scalable as possible, normally one wants to consider graph problems in the \emph{sublinear local space regime}, where local space $\lspace = n^{\spacexp}$ for any given constant $\spacexp \in (0,1)$. (There has been some research considering also the case when $\lspace = \Theta(n)$, or even when $\lspace = n^{1+\spacexp}$ (in which case one wants to study the case that $\lspace \ll m$) but we will not consider such setting in the current paper.) We will require that the number of machines is not significantly more than required, 
specifically that it is $\widetilde{O}(n+\frac{m}{\lspace})$ (note that the optimal amount would be $\widetilde{O}(\frac{n+m}{\lspace})$, but our algorithm requires the ability to assign a machine to each node).
A major challenge in the design of \MPC algorithms in the sublinear local space regime is that the local space of each machine is (possibly) not sufficient to store all the edges incident to a single node. This constraint naturally requires an \MPC algorithm to rely on extensive communication between machines, and most of the techniques known are based on some graph sparsification. It is important to note here that even in the sublinear local space regime, the \MPC model is known \cite{GSZ11} to be stronger than the PRAM model, e.g., it is known that sorting\footnote{Here we consider sorting of $N$ objects on an \MPC with local space $N^{\gamma}$ and on $N^{1-\gamma}$ machines, for any constant $\gamma>0$.} (and in fact, many related tasks, like prefix sum computation) can be performed in a constant number of rounds, even deterministically, see \cite{GSZ11}. Observe that with this tool, we can gather nodes' neighborhoods to contiguous blocks of machines, and learn their degrees, in $O(1)$ rounds, and that we can assume, without loss of generality, that the input can be distributed arbitrarily on the first $\Theta(\tfrac{n+m}{\lspace})$ \MPC machines.


\subsection{{An overview of \LOCAL \DILC algorithm of  Halld\'orsson \etal \cite{hknt_local_d1lc}}}
\label{sec:overview_of_hknt}

In this section, we give a detailed overview of the  work of \cite{hknt_local_d1lc}. Among several other contributions, \cite{hknt_local_d1lc} give an algorithm which, in the \LOCAL model of distributed computing, given an input instance of $(\deg + 1)$-list coloring {for a graph $G$}, colors all nodes of $G$ having degrees at least $\log ^7 n$  in $O(\log^* n)$ rounds. We give an overview of their algorithm here, because our algorithm in \MPC to color all nodes in the degree range $\left[\log^7 n, n^{7\delta}\right]$ uses steps of the algorithm in~\cite{hknt_local_d1lc}, where $\delta$ is set to be sufficiently smaller than $\phi$ (recall that $\phi$ is our local space parameter, i.e., $\lspace=O(n^\phi)$).

The \LOCAL algorithm for \DILC due to Halld\'orsson \etal \cite{hknt_local_d1lc} handles the input graph in ``ranges'' of degree. It begins by coloring vertices with degrees in the range $[\log^7 n, n]$, followed by vertices with degrees in the range $[\log^7\log n, \log^7 n]$ vertices in the range $[\log^7\log\log n, \log^7 \log n]$, and so on, giving $O(\log^* n)$ ranges overall. The coloring algorithm runs in $O(\log^* n)$ rounds for each range.
However, for ranges after the first the algorithm does not color all nodes with high probability, leaving a set of ``bad'' nodes. Some standard ``shattering'' arguments are used (showing that these bad vertices form small components), and then these bad vertices can be colored using a deterministic \DILC algorithm. The post-shattering coloring step on the second degree range ($[\log^7\log n, \log^7 n]$) is the bottleneck in the algorithm, and gives the overall complexity of $\widetilde{O}(\log^2 \log n)$ \LOCAL rounds.~\footnote{The actual bound obtained by \cite{hknt_local_d1lc} is $O(\log ^3 \log n)$, and when this is combined with the result in\cite{GG23}, one gets the round complexity to be $\widetilde{O}(\log ^2 \log n)$.}


The algorithm of \cite{hknt_local_d1lc} (for the degree range $[\log^7 n, n]$) requires several careful arguments which rely on nodes being able to calculate certain parameters, in particular the following:



\begin{definition}[Parameters from \cite{hknt_local_d1lc}]
\label{def:parameters}
The following node-parameters are used in the coloring algorithm of \cite{hknt_local_d1lc}:
\begin{itemize}
\item \emph{Slack} $\slack{v} = p(v) - d(v)$: The difference between a the size of the palette of a node and its degree. 
\item \emph{Sparsity} $\sparsity{v} = \frac{1}{d(v)} \cdot \left[ \binom{d(v)}{2} - m(N(v)) \right]$~\footnote{{$m(N(v))$ denotes the number of edges in the subgraph induced by $N(v)$}.}: The proportion of non-edges between neighbors of $v$. 
\item \emph{Disparity} $\disparity{u}{v} = |\Psi(u) \setminus \Psi(v)| / |\Psi(u)|$: The (proportional) difference between the palettes of $u$ and $v$.
\item \emph{Discrepancy} $\discrepancy{v} = \sum_{u \in N(v)} \discrepancy{u, v}$: The sum of disparities between $v$ and its neighbors.
\item \emph{Unevenness} $\unevenness{v} = \sum_{u \in N(v)} \frac{\max(0, d(u) - d(v))}{d(u) + 1)}$: A quantity that relates to how many of the neighbors of $v$ have much higher degree than $v$.
\item \emph{Slackability} $\slackability{v} = \discrepancy{v} + \sparsity{v}$, \emph{Strong Slackability} $\sslackability{v} = \unevenness{v} + \sparsity{v}$: How easy it is to create slack for a node.
\end{itemize}
\end{definition}

The coloring algorithm \cite{hknt_local_d1lc} (for a single degree range) is outlined in  \Cref{alg:color_middle_degrees}. As described, the algorithm first computes \emph{almost-clique decomposition} (ACD). After computing ACD, each vertex of the graph will be classified as either as a sparse vertex or a uneven vertex or a dense vertex. Then the we call \textsc{ColorSparse} to color the vertices that are either sparse or uneven, and then we call  \textsc{ColorDense} to color the dense vertices. In \Cref{sec:acd}, we discuss how to compute ACD and classify the vertices into sparse and dense vertices. In \Cref{sec:sparse} and \Cref{sec:dense}, we discuss how to color the sparse an dense vertices respectively. There are additional subroutines, called by \textsc{ColorSparse} and \textsc{ColorDence}, that we will discuss from \Cref{sec:acd}--\Cref{sec:dense}.


\begin{algorithm}[H]
\caption{${\text{\sc ColorMiddle}}$: Colors nodes with degrees in $[\log^7 n, n^{7\delta}]$ {\small (extracted from \cite{hknt_local_d1lc})}}
\label{alg:color_middle_degrees}

Compute an almost-clique decomposition (ACD). 

{\sc ColorSparse} (\Cref{alg:color_sparse}).

{\sc ColorDense} (\Cref{alg:color_dense}).
\end{algorithm}


\subsubsection{Computing an almost-clique decomposition}
\label{sec:acd}

First, as already mentioned, the authors of \cite{hknt_local_d1lc} compute an \emph{almost-clique decomposition} of the input graph. The almost-clique decomposition is as follows:

\begin{definition}[{$(\mbox{deg}+1)$-ACD}~\cite{AlonA20, hknt_local_d1lc}]
\label{def:acd}
Let $G=(V,E)$ be a graph and $\veps_{ac}, \veps_{sp} \in (0,1)$ be constants. A partition of $V$ into $V_{\text{sparse}} \sqcup V_{\text{uneven}} \sqcup V_{\text{dense}}$, such that $V_{\text{dense}}$ can be partitioned into $C_1, \ldots, C_t$ for some $t$, is said to be an almost-clique decomposition for $G$ if


\begin{enumerate}
\item[(i)] Every $v \in V_{\text{sparse}}$ is $\veps_{sp}  d(v)$-sparse;
\item[(ii)] Every $v \in V_{\text{uneven}}$ is $\veps_{sp} d(v)$-uneven;
\item[(iii)] For every $i \in [t]$ and $v \in C_i$, $d(v)\leq (1+\veps_{ac})\size{C_i}$;
\item[(iv)] For every $i \in [t]$ and $v \in C_i$, $\size{C_i} \leq (1+\veps_{ac})\size{N(v) \cap C_i}.$
\end{enumerate}
\end{definition}

Intuitively:

\begin{itemize}
\item $V_{\text{sparse}}$: A set of (sufficiently) \emph{sparse} nodes. A node is $\varepsilon$-sparse if there are $\varepsilon \cdot d(v)$ many non-edges between neighbors of $v$.
\item $V_{\text{uneven}}$: A set of \emph{uneven} nodes. A node is uneven if many of its neighbors have a much higher degree than $v$. Formally, a node $v$ is $\eta$-uneven if $\eta_v \geq \sum_{u \in N(v)} \frac{\max(0, d(u) - d(v))}{d(u) + 1}$.

\item $V_{\text{dense}}=C_1\sqcup \dots \sqcup C_t$: Each $C_i$ is an \emph{almost-clique}. Intuitively, almost-cliques are very well connected parts of the graph. Formally, for some $\varepsilon_{sc}$ and for all almost-cliques $C_i$ and all $v \in C_i$, $d(v) \leq (1 + \varepsilon_{ac}) \size{C_i}$, and $\size{C_i} \leq (1 + \varepsilon_{ac}) \size{N(v) \cap C_i}$.
\end{itemize}

After computing an ACD and partitioning the nodes into $V_{\text{sparse}}$, $V_{\text{uneven}}$ and $V_{\text{dense}}$,  nodes in $V_{\text{sparse}} \sqcup V_{\text{uneven}}$ are colored first by \textsc{ColorMiddle} (\Cref{alg:color_middle_degrees}) by subroutine \textsc{ColorSparse} and then the nodes in $V_{\text{dense}}=C_1 \sqcup \dots \sqcup C_t$ are colored hy subroutine \textsc{ColorDense}. 

\subsubsection{Coloring sparse nodes}
\label{sec:sparse}

Before discussing about coloring the sparse nodes, we first discuss about a subroutine  {\sc SlackColor} (\Cref{alg:slack_color}) which is used by Step 2 and 3 of \Cref{alg:color_middle_degrees} by the subroutines \textsc{ColorSparse} and \textsc{ColorDense}, respectively. 

\subsubsection*{SlackColor}
\label{subsubsec:alg-SlackColor}
\textsc{SlackColor} colors all nodes which have slack linear in their degree (and the degree is large enough) in $O(\log^*n)$ rounds. The authors of \cite{hknt_local_d1lc}  point out that similar results were shown in \cite{sw10-distributed-symmetry-breaking}: in this paper we focus on simulating and derandomizing the algorithm {\sc SlackColor} given in \cite{hknt_local_d1lc}.
Note that {\sc SlackColor} takes two parameter $s_{min}$ and $\kappa$ as input where $s_{min}$ is an integer that lies between $1$ and the minimum amount of slack of any node (that we want to color) and $\kappa$ lies between $1/s_{min}$ and $1$. \footnote{The success probability of \textsc{SlackColor} depends on $s_{min}$ and $\kappa$. Also, \cite{hknt_local_d1lc} uses \textsc{SlackColor} only when $s_{min}=\Omega(\ell)$ and $\kappa$ is a suitable constant, where \mbox{$\ell=\log^{2.1}\Delta$ .}}.

\begin{algorithm}[H]

\caption{\textsc{SlackColor}$(s_{min})$, for node $v$, from \cite{hknt_local_d1lc}}
\label{alg:slack_color}

Do ${\text{\sc TryRandomColor}}(v)$ for $O(1)$ rounds.

If $\slack{v} < 2d(v)$ then terminate.


Let $\rho\gets s_{\min}^{1/(1+\kappa)}$. \textit{\small \hfill //$s_{min}$ is globally known such that $1<s_{min}\leq\min_{v}s(v)$ and $\kappa$ is an input parameter such that $1/s_{min}<\kappa \leq 1$.}

\For{$i$ from $0$ to $\log^* \rho$}
{
$x_i \gets 2 \uparrow \uparrow i$ \textit{\small \hfill // {$2 \uparrow \uparrow i$ refers to iterated exponentiation, i.e., $2 \uparrow \uparrow 0 = 1$ and $2 \uparrow \uparrow (i + 1) = 2^{2 \uparrow \uparrow i}$}.}

$\text{{\sc MultiTrial}}(x_i)$ $2$ times

If $d(v) > \slack{v}/\min(2^{x_i},\rho^\kappa)$ then terminate.
}

\For{$i$ from $1$ to $\lceil 1/\kappa \rceil$}
{
$x_i \gets \rho^{i \cdot \kappa}$.

$\text{{\sc MultiTrial}}(x_i)$ $3$ times.

If $d(v) > \slack{v}/\min(\rho^{(i+1) \cdot \kappa},\rho)$ then terminate.
}

$\text{{\sc MultiTrial}}(\rho)$.
\end{algorithm}

{\sc SlackColor} (\Cref{alg:slack_color}) requires two auxiliary subroutines \textsc{TryRandomColor} (\Cref{alg:try_random_color}) and \textsc{MultiTrial} (\Cref{alg:multi_trial}) which we describe below.

\paragraph*{Auxiliary subroutine TryRandomColor}
\label{subsubsec:alg-TryRandomColor}
In {\sc TryRandomColor}, each participating node $v$ selects a color $\psi_v$ from its current color palette uniformly at random. $v$ colors itself with $\psi_v$ if none of the neighbors of $v$, whose colors conflict with $v$, selects $\psi_v$.

\begin{algorithm}[H]
\caption{{\sc TryRandomColor}(node $v$) from \cite{hknt_local_d1lc}}
\label{alg:try_random_color}
Pick $\psi_v$ u.a.r. from $\Psi(v)$.

Send $\psi_v$ to each $u \in N(v)$, receive the set $T = \{\psi_u : u \in N^+(v)\}$, where $N^+(v)$ is the set of neighbours whose colors ``conflict'' with $v$.

If $\psi_v \not\in T$ then permanently color $v$ with $\psi_v$.

Send and receive permanent colors, and remove the received one from $\Psi(v)$
\end{algorithm}

\paragraph*{Auxiliary subroutine MultiTrial}
\label{subsubsec:alg-MultiTrial}
\textsc{MultiTrial} is, informally speaking, a generalized version of \textsc{TryRandomColor} where each participating node selects $x$ number of colors from its palette uniformly at random.

\begin{algorithm}[H]
\caption{$\text{\sc{MultiTrial}}(x)$ from \cite{hknt_local_d1lc}}
\label{alg:multi_trial}
$v$ picks a set $X_v$ of $x$ random colors in its palette $\Psi(v)$, sends them to neighbours.

\If{$\exists \psi \in X_v$ such that $\forall u \in N(v), \psi \not \in X_v$}
{
Permanently color $v$ with some such $\psi$ and broadcast to $N(v)$.
}
\end{algorithm}

\subsubsection*{Coloring sparse or even nodes}
\label{subsubsec:ColorSparse}

The first coloring task is to color all nodes which are \emph{sparse} or \emph{uneven} (that are in $V_{\text{sparse}}$ or $V_{\text{uneven}}$). The outline of the subroutine which does this {\sc ColorSparse} (\Cref{alg:color_sparse}). This is a recursively subroutine that uses an auxiliary subroutine {\sc GenerateSlack} (\Cref{alg:generate_slack}).

\begin{algorithm}[H]
\caption{{\sc ColorSparse} from \cite{hknt_local_d1lc}}
\label{alg:color_sparse}
Identify set of nodes $V_{\text{start}} \subset V_{\text{sparse}}$.

{\sc GenerateSlack} in $G\left[(V_{\text{sparse}} \cup V_{\text{uneven}}) \setminus V_{\text{start}}\right]$.

{\sc SlackColor} $V_{\text{start}}$.

{\sc SlackColor} $V_{\text{sparse}} \setminus V_{\text{start}}$ and $V_{\text{uneven}}$.
\end{algorithm}

\paragraph*{Auxiliary subroutine GenerateSlack}
\label{subsubsec:alg-GenerateSlack}

The challenge in {\sc ColorSparse} (\Cref{alg:color_sparse}) is that all nodes at the start of the algorithm are only guaranteed to have constant slack: in the (easier) $(\Delta + 1)$-coloring problem, nodes of low-degree are guaranteed to have slack linear in their degree. The authors of \cite{hknt_local_d1lc} overcome this by identifying some subset of the nodes in $V_{\text{sparse}}$ for which it is hard to generate slack: they call this set $V_{\text{start}}$. They then generate some slack using a simple primitive ({\sc GenerateSlack}, \Cref{alg:generate_slack}) for nodes in $(V_{\text{sparse}} \cup V_{\text{uneven}}) \setminus V_{\text{start}}$.

\begin{algorithm}[H]
\caption{{\sc GenerateSlack} from \cite{hknt_local_d1lc}}
\label{alg:generate_slack}
$S \gets$ each node $v$ is sampled into $S$ independently with probability $\frac{1}{10}$.

For all $v \in S$ in parallel $\text{{\sc TryRandomColor}}(v)$.
\end{algorithm}

The authors of \cite{hknt_local_d1lc} show that nodes in $V_{\text{start}}$ have many neighbours in $(V_{\text{sparse}} \cup V_{\text{uneven}}) \setminus V_{\text{start}}$, and so nodes in $V_{\text{start}}$ have considerable \emph{temporary slack} (slack obtained from neighbors being colored later). Since $(V_{\text{sparse}} \cup V_{\text{uneven}}) \setminus V_{\text{start}}$ are exactly the nodes which do not struggle to obtain slack, they have that all sparse/uneven nodes have slack linear in their degree, and this is enough to color them in $O(\log^* n)$ rounds (by the algorithm of \cite{sw10-distributed-symmetry-breaking} or by {\sc SlackColor}).

\subsubsection{Coloring dense nodes}
\label{sec:dense}

Finally, it remains to color nodes in the almost-cliques $C_1 \sqcup \dots \sqcup C_t$. This task is solved using algorithm {\sc ColorDense} (\Cref{alg:color_dense}) \cite{hknt_local_d1lc}.

\begin{algorithm}[H]
\caption{{\sc ColorDense} from \cite{hknt_local_d1lc}}
\label{alg:color_dense}

Compute the leader $x_C$ and outliers $O_C$ for each almost-clique $C$. Let $O = \bigcup_{C \in \cC} O_C$, where $\cC$ denotes the set of all almost-cliques.

{\sc GenerateSlack}.

$P_C \leftarrow${\sc PutAside}$(C)$ for each low slack almost-clique $C$. Let $P$ denote the set of union of all such $P_C$'s.

{\sc SlackColor} $O$.

{\sc SynhColorTrial} $V _{\text{dense}} \setminus P$.

{\sc SlackColor} $V_{\text{dense}} \setminus P$.

For each $C \in C$ with low slackability, $x_C$ collects the palettes in $P_C$ and colors the nodes locally.
\end{algorithm}

First, each almost-clique $C$ selects a \emph{leader} $x_C$, i.e., the node with minimum slackability (breaking ties arbitrarily).~\footnote{See Definition~\ref{def:parameters} for the definition of slackability of a node.} The slackability of an almost-clique $C$ is the slackability of the leader $x_C$. $C$ is said to an almost-clique with low slackability if its slackability is at most $\ell=\log^{2.1} \Delta$; otherwise $C$ has high slackability. For each each almost-clique with low slackability, we need to to remove some of the nodes to help the coloring process. The details will be discussed later in this section.

The intuition is that the leader has similar palettes to a large fraction of other nodes in the almost-clique, a property which will be useful later. Then the almost-clique is split roughly in half, into \emph{inliers} and \emph{outliers}: inliers are the nodes that are similar to the leader (by some criteria that we will discuss later), and outliers are nodes which are dissimilar to the leader. In particular, since outliers are (slightly less than) half of the almost-clique, they have slack linear in their degree and can be colored in $O(\log^* n)$ rounds by {\sc SlackColor} (if we defer the nodes in the inliers to be colored later).

Next, inliers need to be colored. Since the leader is very similar to the inliers (in particular: it is connected to all inliers and has a similar palette to them), a surprisingly simple process, where the leader proposes one of its colors to each inlier (such that it proposes each color to at most one inlier), suffices. The pseudocode for this process is given in {\sc SynchColorTrial} (\Cref{alg:synch}). After {\sc SynchColorTrial} has completed, there is linear slack at each remaining inlier to apply {\sc SlackColor} and color the remaining inliers in $O(\log^* n)$ rounds.

\begin{algorithm}[H]
\caption{{\sc SynhColorTrial} for almost-clique $C$ \cite{hknt_local_d1lc}}
\label{alg:synch}
{

$x_C$ randomly permutes its pallete $\Psi(x_C)$, sends each neighbor $v \in I_C$ a distinct color $c(v)$.
\textit{\small \hfill // $I_C=C\setminus O_C$ denotes the set of inliers of the almost-clique $C$.}

Send $c(u)$ to $N(v)$, and receive the set $T=\{c(u):u \in N^{+}(v)\}$. \textit{\small \hfill // $N^+(v)$ is the set of neighbours whose colors ``conflict'' with $v$.}

If $c(v) \notin T$, then permanently color $v$ with $c(v)$.

Send and receive permanent colors, and remove the received one from $\Psi(v)$.
}
\end{algorithm}

Finally, there is an additional challenge to overcome for some almost-cliques. If the ``slackability'' (intuition: a measure of easy it is to generate slack) is low, {\sc SynchColorTrial} might not create enough slack for the remaining inliers. The authors of \cite{hknt_local_d1lc} get around this by computing a ``put-aside'' set: a set of inliers which we set aside and color at the end of the algorithm. The put-aside set is computed using {\sc PutAside} (\Cref{alg:put-aside}), and has three important properties: it is of polylogarithmic size (in each almost-clique that needs it); put-aside sets in different almost-cliques have no edges between them; and the inliers not assigned to the put-aside set each have {sufficient neighbours in the put-aside sets} which are to create sufficient slack to be colored ultrafast by {\sc SlackColor}.

\begin{algorithm}[H]
\caption{{\sc PutAside}($C$) from~\cite{hknt_local_d1lc}}
\label{alg:put-aside}
{
Sample each node in $v \in I_C$ independently with probability $p_s=\frac{\ell^2}{48 \Delta_C}$, and include in the set $S_C$, where $\ell=\log ^{2.1} \Delta$.

Let $S$ denote the union of such $S_C$'s.

return $P_C$ as the set of nodes in $S_C$ that does not have a neighbor in $S$.
}
\end{algorithm}

\section{Randomized \DILC Algorithm in \MPC}
\label{sec:randomized_d1lc_intuition}

In this section, we  argue that if the maximum degree of the input instance is not too large, then, when combined with a result of Czumaj \etal \cite{CDP21}, the \LOCAL algorithm, for \DILC due to Halld\'orsson \etal \cite{hknt_local_d1lc}, can be efficiently implemented in exponentially fewer rounds in the \MPC model. Since this \emph{randomized} implementation is rather straightforward, we will only sketch it here. However, it is not at all clear that the algorithm of  Halld\'orsson \etal \cite{hknt_local_d1lc} can be efficiently derandomized in \MPC. The main contribution of this paper is a derandomization framework (presented in \Cref{sec:derandomization-framework}) and its use for a deterministic \MPC algorithm for \DILC.

\hide{
In this section, we give an overview of the \LOCAL algorithm for \DILC of Halld\'orsson \etal \cite{hknt_local_d1lc}. We argue that it can be efficiently implemented in exponentially fewer rounds in the \MPC model, when combined with a result of Czumaj \etal \cite{CDP21}, provided that the maximum degree of the input instance is not too large. We will give a more detailed overview of the work in \Cref{sec:overview_of_hknt}.

\sam{The description of \cite{hknt_local_d1lc} omits a lot of details, but perhaps it's enough for the reader to be getting on with?}
}

\label{subsec:randomized-MPC}
Recall that the \LOCAL algorithm for \DILC due to Halld\'orsson \etal \cite{hknt_local_d1lc} handles the input graph in ``ranges'' of degree. It begins by coloring vertices with degrees in the range $[\log^7 n, n]$, followed by vertices with degrees in the range $[\log^7\log n, \log^7 n]$ vertices in the range $[\log^7\log\log n, \log^7 \log n]$, and so on, giving $O(\log^* n)$ ranges overall. Here, we consider the implementation of the (randomized) algorithm of \cite{hknt_local_d1lc} on a single range of degrees in the low-space \MPC model. We will only sketch the approach here (since this is not our main result): we give this partial result to build intuition for which aspects of \Cref{thm:main:deter} are challenging to obtain. 


\comments{Inspecting the algorithm presented in \cite{hknt_local_d1lc} unveils that, in the worst case,} each of its steps either require a node $v$ to send messages to all its neighbors (of size $O(\deg(v))$), or to calculate some value or send some messages based on the \comments{full content} of the $2$-hop neighborhood of $v$.
This requires $\lspace \geq \Delta^2$ (or equivalently, for a fixed \lspace we need to have $\Delta \leq \sqrt{\lspace}$), since a single machine has to be able to store $\deg(v)$ messages of size $O(\deg(v))$, and the $2$-hop neighborhood of any vertex. Note also that the global space required is $O(m+n^{1+\phi})$. 

If we have this amount of space relative to our maximum degree, then we can implement $1$ round of the \LOCAL algorithm of \cite{hknt_local_d1lc} in $O(1)$ rounds of low-space \MPC. Like the \LOCAL algorithm, our \MPC implementation would succeed with high probability for all nodes with degrees in $[\log^7 n, \sqrt{\lspace}]$.

Also, as mentioned earlier, there is already an algorithm in low-space \MPC (see, e.g., Czumaj \etal \cite{CDP21}) which colors \DILC instances with polylogarithmic maximum degree in $O(\log\log\log n)$ rounds deterministically.
Combining these two observations yields the following:

\begin{lemma}
\label{obs:partial_d1lc}
    Let $\phi \in (0, 1)$ be an arbitrary constant. There exists an algorithm that, for every $n$-node graph $G=(V, E)$ with maximum degree $\Delta = O(\sqrt{n^\phi})$, solves the \DILC problem using $O(\log\log\log n)$ rounds with high probability, in the sublinear local space \MPC model with local space $O(n^\phi)$ and global space $O(m+n^{1+\phi})$.
\end{lemma}

Two enhancements are required to get from \Cref{obs:partial_d1lc} to our main result, \Cref{thm:main:deter}.
In \Cref{sec:derandomization-framework}, we give our framework for derandomizing coloring algorithms in \MPC, and in \Cref{sec:deterministic_d1lc} we show that we can derandomize the implementation of the algorithm of \cite{hknt_local_d1lc} for the degree range $[\log^7 n, {\lspace}^c]$ for some suitable constant $c \in (0,1/2)$.
Finally, in \Cref{sec:d1lc_for_all_degree_ranges} we show that we can, \emph{deterministically}, reduce any input instance to a collection of instances with lower degree. We ensure that the palettes of these instances are mostly distinct, so that only a constant number of instances need to be colored sequentially.

\section{Framework to Derandomize \MPC Algorithms}
\label{sec:derandomization-framework}

In this section, we present our black-box derandomization technique for coloring problems in \MPC. Our framework relies on a collection of suitable properties (normal distributed procedures, see \Cref{prng-derandomizable}) that allow us to fully derandomize randomized algorithm using a combination of pseudorandom generators (PRGs) and the method of conditional expectations.


\subsection{Normal distributed procedures}
\label{subsec:normal-distributed-procedures}

Our framework relies on the notion of normal $(\nr,\Delta)$-round randomized distributed procedure. This notion is introduced to combat the issue that a PRG that fits on a machine in low-space \MPC causes more nodes to fail than the underlying random process it is applied to. We wish to defer these failed nodes to deal with them later, but in some procedures this could cause a chain reaction of failures and result in an unsolvable instance. So, \Cref{prng-derandomizable} captures those procedures for which we will show we can safely defer failed nodes.

\Cref{prng-derandomizable} is quite technical, but is conceptually fairly simple, so we first explain its meaning before giving the formal notation. In general, the randomized distributed procedures we are interested in are algorithms that run for some number of rounds in a distributed model (in our case, \LOCAL), and for which it is shown that upon termination, a particular desirable property (which we call the ``strong success property'') holds for all nodes with high probability. In coloring problems, this property could be, for example, that all nodes have sufficient \emph{slack} (difference between remaining palette size and number of uncolored neighbors). More generally, this could include properties such as having low remaining degree (i.e. few neighbors that have not yet terminated and given their final output for the graph problem of interest).

The defining property of a \emph{normal} distributed procedure is that we are able to define a ``weak success property'' that is still sufficient for the overall algorithm to proceed, and is implied by the strong success property \emph{even if an arbitrary subset of nodes are deferred, changing their outputs from the procedure to the special \textsc{Defer} marker}. That is, if we were to run the procedure, causing all nodes to satisfy the strong success property (which will happen with high probability), and then an adversary were able to defer any subset of nodes, effectively nullifying their outputs, we would still be sure to satisfy the weak success property at all nodes. In this case, deferring does not substantially disrupt the overall algorithm, since if we continue to repeat the procedure on the deferred nodes, we will eventually reach a state in which all nodes have the required property to proceed to the next step.

This may seem like a very strong requirement for a randomized distributed procedure: indeed, it can be seen as a special case of the distributed Lov\`asz Local Lemma, and much research has been devoted to the more general case where deferring nodes can cause their neighbors to no longer satisfy any success property (see e.g. \cite{GHK18, Davies23}). However, as we will see, it applies to many useful procedures, and particularly those for coloring problems, where deferring nodes is almost always helpful and provides slack to neighbors.

\begin{definition}
\label{prng-derandomizable}
A \textbf{normal $(\nr,\Delta)$-round distributed procedure} running on a graph $G$, of maximum degree \emph{at most} $\Delta$, is a procedure in the randomized \LOCAL model satisfying the following criteria:

\begin{itemize}
\item The procedure takes $\nr$ rounds of \LOCAL.
\item At the beginning of the procedure, nodes $v$ have $O(\Delta^{\nr})$-word sets of input information $\textsc{In}_v$ associated with them.
\item During the procedure, nodes only use information from their $\nr$-hop neighborhood (i.e., from $\textsc{In}_v \cup \bigcup_{u\in N^{\nr}(v)} \textsc{In}_u\}$ and $O(\Delta^{2\nr})$ random bits, and perform $O(\Delta^{8\nr})$ computation. \footnote{$N^{\nr}(v)$ denotes the set of vertices that are with in $\nr$-hop neighborhood of $v$.}
\item The output of the procedure is a new $O(\Delta^{\nr})$-word output information $\textsc{Out}_v$ for each node, from a set of possible outputs \textsc{Out}.
\item The procedure has a \textbf{``strong success property''} (computable with $O(\Delta^{8\nr})$ computation) that determines whether it has been successful for a particular node, based on the output information of that node's $\nr$-hop neighborhood (formally, $ \textsc{SSP}_v : \comments{\textsc{Out}^{N^{\nr}(v)}} \rightarrow \{T,F\}$).
\item At the end of the procedure, for any node $v$,
\begin{align*}
    \Prob{\textsc{SSP}_v\left(\prod_{u\in N^{\nr}(v)} \textsc{Out}_u \right)} &\ge 1-\frac{1}{2n}
    \enspace.
\end{align*}
\item The procedure also has a \textbf{``weak success property''} (also computable with $O(\Delta^{8\nr})$ computation) that extends the output domain of each node to include a special \textsc{Defer} marker (not to be conferred by the procedure itself) indicating that that node will be deferred until the end of the derandomization (formally, $\textsc{WSP}_v: \{\textsc{Out}\cup \textsc{Defer}\}^{N^{\nr}(v)} \rightarrow \{T,F\} $).
\item Denote the set of nodes that do not satisfy the strong success property as $\overline{\textsc{SSP}}$. The success properties are such that if a node $v$ satisfies $\textsc{SSP}_v$, and nodes in $\overline{\textsc{SSP}}$ are deferred, $v$ still satisfies the weak success property. Formally:
\begin{align*}
    & \textsc{SSP}_v\left(\prod_{u\in N^{\nr}(v)} \textsc{Out}_u\right) 
    \implies  
     \textsc{WSP}_v\left(\prod_{u\in N^{\nr}(v)\setminus \overline{\textsc{SSP}}} \textsc{Out}_u \times \prod_{u\in N^{\nr}(v)\cap \overline{\textsc{SSP}}} \textsc{Defer}\right)
    \enspace.
\end{align*}
\end{itemize}
\end{definition}

Let us now discuss how \Cref{prng-derandomizable} can be applied. The intuition behind \Cref{prng-derandomizable} is that it captures procedures whose output properties are not damaged too much by unsuccessful nodes deferring. As an example, consider Luby's randomized algorithm for maximal independent set \cite{Luby86}. Luby's algorithm, in $O(\log n)$ rounds of \LOCAL, produces an output set that is certainly independent, and is maximal with high probability. We can define both success properties (strong and weak) for a node $v$, to be that $v$ has a node within distance $1$ in the resulting output set. Note that this only captures maximality, not independence, but that is sufficient since independence is guaranteed by the process, and a ``failed'' run only violates maximality. A standard analysis of Luby's algorithm would give, as required, that $\Prob{\textsc{SSP}_v \left(\prod_{u\in N^{\nr}(v)} \textsc{Out}_u \right)} \ge 1 - \frac{1}{2n}$, i.e., that a node $v$ is successful with high probability. Notice also that only nodes that are not in the output independent set can fail to satisfy \textsc{SSP}, by definition. Therefore, deferring such nodes does not remove any nodes from the output independent set, and so all nodes satisfying $\textsc{SSP}_v$ also satisfy $\textsc{WSP}_v$. 

The arguments above imply that Luby's algorithm is a normal $(O(\log n),\Delta)$-round distributed procedure (the other necessary properties are trivial to check). In fact, we had something stronger: we used essentially the same condition for the strong and weak success property (with the only formal difference being that the weak success property has a domain extended to include \textsc{Defer} tags). This indicates that, in this case, deferring unsuccessful nodes does not hurt us \emph{at all}. This is something we will also see in our main application to \DILC. \Cref{prng-derandomizable} is more general, though, and allows $\textsc{WSP}$ to be weaker than $\textsc{SSP}$ in addition to incorporating \textsc{Defer} tags, which can provide some leeway in derandomization.

We will argue (in \Cref{sec:derand}) that we can only efficiently derandomize normal $(\nr,\Delta)$-round distributed procedures in \MPC when $\nr$ is low (ideally constant). This means that, in most applications, it will not suffice to show that entire algorithms for a problem are normal distributed procedures; we must instead consider those algorithms as series of short subroutines, and then show that each of those subroutines is a normal $(\nr,\Delta)$-round distributed procedure. This means defining success properties that capture an appropriate level of progress after each subroutine, not just correctness of the final output. Indeed, the ability to capture subroutines in this way is one of the main generalizations of our approach over that of \cite{CDPcompstab}.

Coloring algorithms are particularly amenable to this type of analysis, since they often consist mostly of subroutines to generate \emph{slack}. That is, some short subroutine is performed, after which analysis shows that nodes' remaining palette sizes will be \comments{larger than their remaining degrees} by some amount. This measure of slack is then the success property for the subroutine. Importantly, deferring nodes (until the end of the entire coloring process) can only ever increase slack, since those deferred nodes are removed from neighbors' neighborhoods, but do not block any colors from neighbors' palettes. So, we can again define both $\textsc{SSP}_v$ and $\textsc{WSP}_v$ to be essentially the same property: that node $v$ has at least some amount of slack at the end of a subroutine. As in the example of Luby's algorithm, deferring unsuccessful nodes (in fact, in this case, even successful nodes) cannot hurt $v$ at all, and $\textsc{SSP}_v\implies \textsc{WSP}_v$ under any subset of deferrals.

Next, we discuss some known results on pseudorandom generators that we use in our derandomization framework.

\subsection{PRGs and derandomization}
\label{sec:derand}

A \emph{Pseudorandom Generator (PRG)} is a function that takes a short \emph{random seed} and produces a longer string of \emph{pseudorandom} bits, which are computationally indistinguishable from truly random bits. We use the following definition from \cite{Vadhan12} for indistinguishability:

\begin{definition}
[Definition 7.1 in \cite{Vadhan12}]
\label{def:computational-indistinguishability}
Random variables $X$ and $Y$ taking values in $\{0,1\}^m$ are $(t,\eps)$ \textbf{indistinguishable} if for every non-uniform algorithm $T:\{0,1\}^m\rightarrow \{0,1\}$ running in time at most $t$, we have $|\Pr[T(X)=1]-\Pr[T(Y)=1]| \le \eps$.
\end{definition}

Let $U_k$ denote a random variable generated uniformly at random from $\{0,1\}^k$. Then pseudorandom generators are defined as follows:

\begin{definition}[PRG, Definition 7.3 in \cite{Vadhan12}]
\label{def:PRG}
A deterministic function $\mathcal{G}:\{0,1\}^d \to \{0,1\}^m$ is an \textbf{$(t,\eps)$ pseudorandom generator (PRG)} if (1) $d < m$, and (2) $\mathcal{G}(U_d)$ and $U_m$ are $(t,\eps)$ indistinguishable.
\end{definition}

A simple application of the probabilistic method can show the existence of PRGs with optimal parameters:

\begin{proposition}[Proposition 7.8 in \cite{Vadhan12}]
\label{prop:perfect-PRG}
For all $t \in \mathbb{N}$ and $\eps>0$, there exists a (non-explicit) $(t,\eps)$ pseudorandom generator $\mathcal{G}: \{0,1\}^d \to \{0,1\}^t$ with seed length $d = \Theta(\log t+\log(1/\eps))$.
\end{proposition}

As shown in \cite{CDPcompstab}, such a PRG can be computed using relatively low space (but exponential computation).

\begin{lemma}[Lemma 35 of \cite{CDPcompstab}, arXiv version]
\label{lem:prg-alg}
For all $t \in \mathbb{N}$ and $\eps>0$, there exists an algorithm for computing the $(t,\eps)$ PRG of \Cref{prop:perfect-PRG} in time $\exp(poly(t/\eps))$ and space $poly(t/\eps)$.
\end{lemma}

Next, we show how to use PRGs to derandomize normal distributed procedures.

\subsection{Derandomizing normal distributed procedures}
\label{sec:normalderand}

In this section, we will sometimes be working on graphs $G$ which are smaller induced subgraphs of the original graph. So, we use $n_G$ to denote the number of nodes in $G$, as opposed to $n$ which will denote the original number of nodes. This distinction is primarily because the space bounds of the \MPC machines are still in terms of $n$ even when working on a smaller graph $G$. Similarly, for a node $v$, $N_G(v)$ \comments{denotes} the set of neighbors of $v$ in $G$ and $d_G(v)$ denotes the degree of $v$ in $G$.

For any integer $\ell\in \mathbb{N}$, let $G^\ell$ be the $\ell$-power of $G$, that is, $G^\ell$ is the graph having the same node set as $G$ and any pair of nodes at a distance at most $\ell$ in $G$ form an edge in $G^\ell$.

We begin with the following lemma.

\begin{lemma}
\label{lem:phasederand}
There is a constant $C$ such that, given an $O(\Delta^{8\nr})$-coloring of $G^{4\nr}$, any normal $(\nr,\Delta)$-round distributed procedure on a graph $G$ can be derandomized in $O(\nr)$ rounds of \MPC, using $\lspace=O(\Delta^{\nr C})$ local space per machine and global space $O(n_G \Delta^{\nr C})$, with the following properties:
\begin{itemize}
\item At most $(\frac 12+ \Delta^{-11\nr})n_G$  nodes are deferred.
\item All non-deferred nodes $v$ satisfy the weak success property.
\end{itemize}
\end{lemma}
\begin{proof}
First, for each node $v$ we collect the input information of its $8\nr$-hop neighborhood ($\textsc{In}_v \cup \bigcup_{u\in N^{\nr}_G(v)} \textsc{In}_u$) to a dedicated machine. This takes $\nr$ rounds, and requires $O(\Delta^{8\nr} \cdot \Delta^{\nr}) = O(\Delta^{11\nr})$ space per machine and $O(n_G \Delta^{11\nr})$ global space. Our aim is then to simulate the procedure using randomness produced by the $(\Delta^{11\nr}, \Delta^{-11\nr})$ PRG implied by \Cref{prop:perfect-PRG}. This PRG has seed length $d = \Theta(\log \Delta)$ and requires $poly(\Delta^{11\nr})$ space to construct and store. We choose $C$ so that this \comments{$poly(\Delta^{22\nr})$} term is $O(\Delta^{\nr C})$.

The PRG, when evaluated on a seed, produces a string of $\Delta^{11\nr}$ pseudorandom bits. We use the provided $O(\Delta^{8\nr})$-coloring of $G^{4\nr}$ to split this string into the input randomness for each node. By definition of a normal $(\nr,\Delta)$-round distributed procedure, each node requires $O(\Delta^{2\nr})$ random bits, and we provide a node colored $i$ in the $O(\Delta^{8\nr})$-coloring with the $i^{th}$ chunk of $O(\Delta^{2\nr})$ bits from the PRG's output. This means that any pair of nodes within distance $4\nr$ receive disjoint chunks of pseudorandom bits.

The output of the PRG under a random seed is $(\Delta^{11\nr},\Delta^{-11\nr})$ indistinguishable from a uniform distribution. Consider the process of simulating the procedure for all nodes within distance $\nr$ of a node $v$, and then evaluating the strong success property $\textsc{SSP}_v(\bigcup_{u\in N^{\nr}_G(v)} \textsc{Out}_u) $. 
By definition of a normal $(\nr,\Delta)$-round distributed procedure, this combined process requires $O(\Delta^{9\nr})$ computation, and depends on the input information and randomness of nodes up to distance $2\nr$ from $v$ (and note that all nodes within this radius receive different chunks of the PRG's output as their pseudorandom bits). This combined process can therefore be run on one \MPC machine. Furthermore, it can be considered a non-uniform algorithm using at most $\Delta^{9\nr}$ computation, and so is `fooled' by the PRG. This means that the output ($T$ or $F$, indicating whether the success property is satisfied) at $v$ differs with probability at most $\Delta^{-11\nr}$ from what it would be under full randomness. That is, the output will be $F$ with probability at most $\frac{1}{2n_G} + \Delta^{-11\nr}$. 

The expected number of nodes \comments{that do not} satisfy the strong success property, when simulating the procedure using the PRG with a random seed, is therefore at most $\frac 12+n_G \Delta^{-11\nr}$. As this value is an aggregate of functions computable by individual machines, using the method of conditional expectations (as implemented for low-space \MPC in \cite{CDPsparse,CDP20}) we can deterministically select a seed for the PRG for which the number of nodes which do not satisfy the success property is at most its expectation (i.e., $\frac 12+n_G \Delta^{-11\nr}$), in $O(1)$ rounds.

Then, we simply mark the nodes which do not satisfy the strong success property as deferred. By \Cref{prng-derandomizable}, all non-deferred nodes still satisfy the weak success property. We therefore meet the conditions of the lemma.
\end{proof}

We can now iterate \Cref{lem:phasederand} in such a way that we can derandomize full algorithms for problems. To do so, we need these problems to satisfy a \emph{self-reducibility property}, which informally means that partial solutions to the problem should be \comments{extendable} to full solutions. 
As discussed, \DILC has this property, which is the reason why it emerges naturally in algorithms for $(\Delta+1)$-coloring or $(\Delta+1)$-list coloring, which themselves are not \mbox{self-reducible}.

\begin{definition}\label{def:selfred}

Consider a graph problem $\mathcal P$ in which \comments{each graph node $v$ takes} some input information $\textsc{IN}_v$ and must produce some output labelling. We call $\mathcal P$ \textbf{self-reducible} if it has the following property:

\begin{itemize}
\item For any graph $G$, any subset $S\subseteq V(G)$, and any valid output labelling on $G$, nodes $v\in S$ can compute in $O(1)$ rounds of \LOCAL new input information $\widehat{\textsc{IN}}_v$ such that:
\begin{itemize}
\item the graph induced on $S$, with new inputs $\widehat{\textsc{IN}}$, forms a valid instance of $\mathcal P$, and
\item replacing the output labelling of nodes in $S$ with any valid output labelling for this induced problem still forms a valid output of the original problem on $G$.
\end{itemize}
\end{itemize}
\end{definition}
Now, we can now state our main derandomization theorem.

\begin{theorem}
\label{lem:fullderand}
Consider a randomized algorithm $\mathcal A$ for a self-reducible problem $\mathcal P$ which consists of a series of $k$ (\comments{$k=n^{o(1)}$}) normal $(\nr,\Delta)$-round distributed procedures (i.e., the final weak success property implies a valid output for $\mathcal P$). Then, there is a constant $C$ such that for any $\delta \in (0,1)$, on any graph $G$ with $\Delta \le n^{7\delta}$, $\mathcal A$ can be derandomized in $O(k\nr + \log^* n_G)$ rounds of \MPC, using $\lspace = O(n^{7\delta \nr C})$ space per machine and global space $O(n_G \cdot n^{7\delta \nr C})$, with all nodes giving a valid output for $\mathcal P$. $\mathcal A$ may also contain up to $O(k\nr + \log^* n_G)$ deterministic steps of \LOCAL or \MPC adhering to the same space bounds.
\end{theorem}

\begin{proof}
We first note that any deterministic \LOCAL or \MPC steps contained in $\mathcal A$ can simply be run as they are, and do not need any derandomization. For derandomization of the randomized procedure, we begin by computing an $O(\Delta^{8\nr})$-coloring of $G^{4\nr}$ in $O(\nr+ \log^* n)$ rounds, by simulating round-by-round the $O(\Delta^2)$-coloring algorithm of Linial\cite{Linial92} on the graph $G^{4\nr}$. Constructing this graph requires collecting $4\nr$-radius balls around each node onto machines, which can be done since we allow $O(n^{7\delta \nr C})$ space per machine.

Then, we apply \Cref{lem:phasederand} to derandomize each of the $k$ normal $(\nr,\Delta)$-round distributed procedures forming $\mathcal A$ in order, in each case using $\Delta = n^{7\delta}$ as an upper bound on the maximum degree (even though the actual maximum degree may be significantly lower). The result is that at most $k(\frac 12 + n_G \Delta^{-11\nr})$ nodes are deferred, and the other nodes satisfy the weak success property of the final procedure (have a valid output labelling). This takes $O(k\nr)$ rounds of \MPC, and uses $\lspace = O(n^{7\delta \nr C})$ space per machine and global space $O(n_G \cdot n^{7\delta \nr C})$.

Next, since the problem is self-reducible, the deferred nodes can compute, in $O(1)$ rounds, new inputs such that it suffices to solve the problem on the induced graph of deferred nodes with these new inputs. To do so, we recursively apply the above process, again using $\Delta = n^{7\delta}$ as an upper bound on the maximum degree. After $r$ of these recursive applications, the number of remaining deferred nodes is at most $\frac k2 + n(k^r \cdot n^{-11\nr r\cdot7\delta})$. Taking $r=1/\delta$ (which is $O(1)$) and $\nr$ to be at least a suitable large constant, we then have $n^{o(1)}$ remaining deferred nodes. We can greedily find valid output labels for them in $O(1)$ further rounds collecting their induced graph onto a single machine, which greedily assigns them valid output labels in any order. Due to the self-reducibility of the problem, we now have a valid output for all nodes for the original problem on $G$.

The total number of rounds used is $O(\nr + \log^* n + kr\nr) = O(k\nr+\log^* n)$, and the space used is $\lspace = O(n^{7\delta \nr C})$ per machine and $O(n_G \cdot n^{7\delta \nr C})$ global space.
\end{proof}

\section{\DILC Algorithm when degree is at most $n^{7\delta}$}
\label{sec:deterministic_d1lc}

In this section we show that our framework from \Cref{sec:derandomization-framework}, when combined with the \DILC algorithm of Halld\'orsson \etal \cite{hknt_local_d1lc}, leads to efficient deterministic \MPC algorithm for \DILC when the maximum degree is $n^{7\delta}$. Note that $\delta \in (0,1)$ is a constant set suitably low relative to the local space parameter $\phi$ and is to be fixed later when we finally prove \Cref{thm:main:deter} in \Cref{sec:d1lc_for_all_degree_ranges}. Specifically,  we show that in $O(\log^* n)$ rounds, with $\lspace=O(n^\phi)$ and $O(n^{1+\phi})$ global space, a subset of the vertices are colored deterministically and the graph induced on the nodes left uncolored will have maximum degree $O(\log^7 n)$ (and will be a \DILC instance, due to the self-reducibility of \DILC). In \Cref{sec:derand}, we argue
how to derandomize the randomized subroutines of  Halld\'orsson \etal, using the framework discussed in \Cref{sec:derandomization-framework}, when the maximum degree of the graph is at most $n^{7\delta}$. In \Cref{sec:randomized-d1lc}, we will argue that the deterministic subroutines which the algorithm of Halld\'orsson \etal uses can be implemented in sublinear \MPC, provided the maximum degree is at most $\sqrt{\lspace}$. Due to a result of Czumaj \etal \cite{CDP21}, this graph can be colored in $O(\log\log\log n)$ rounds.



\hide{In this section, we apply our framework from \Cref{sec:derandomization-framework} to design a deterministic \MPC algorithm for \DILC in a graph $G$ with maximum degree at most $n^{7\delta}$, which takes $O(\log \log \log n)$ rounds on an \MPC with $\lspace=O(n^\phi)$ and $O(n^{1+\phi})$ global space~\footnote{The number of edges is always at most $O(n^{1+\phi})$.}. This is possible due to the simulation of  deterministic steps in \cite{hknt_local_d1lc} when the maximum degree is at most $\sqrt{\lspace}$ as discussed in \Cref{sec:randomized-d1lc} and the fact that the randomized steps in \cite{hknt_local_d1lc} can be efficiently derandomized. In particular, we will discuss the derandomization of \textsc{TryRandomColor}(\Cref{alg:try_random_color}), \textsc{SlackColor} (\Cref{alg:slack_color}), {\sc GenerateSlack} (\Cref{alg:generate_slack}), \textsc{PutAside} (\Cref{alg:put-aside}) and \textsc{SyncColorTrial} (\Cref{alg:synch}) by proving that all of them are $(O(\log ^*n),\Delta)$-round  distributed procedures and arguing that such procedures can be efficiently derandomized by using pseudorandom generators.  We can show that, in one iteration with $O(\log^* n)$ rounds, the number of uncolored nodes will be reduced by a factor of $n^{\Omega(1)}$ (that are deferred to be colored in the later iterations). So, by repeating the process for $O(1)$ iterations, we have $O(\log n)$ nodes left to be colored. Another thing to note that, we are considering the algorithm of \Cref{sec:overview_of_hknt} here when the degree of the nodes is at least $\log ^7 n$. So, we always defer the coloring of the nodes whose degree is at most $\log ^7 n$. So, at the end we have a graph (of at most $n$ nodes) having maximum degree $\log ^7 n$. Note that this can be colored in $O(\log \log \log n)$ rounds due to a result of \cite{CDP21}.}


In this section, graph $G$ \comments{is often} not clear from the context. So, we use $n_G$ to denote the number of nodes in $G$. For a node $v$, $N_G(v)$ denote the set of neighbors of $v$ in $G$ and $d_G(v)$ denotes the degree of $v$ in $G$. We denote the maximum degree of any node in $G$ by $\Delta_G$. For a node $v$ and $\nr \in \mathbb{N}$, $N^{\nr}_G(v)$ denotes the nodes in the \nr-hop neighborhood of $v$.

\subsection{Derandomization of key randomized subroutines of \cite{hknt_local_d1lc}}\label{sec:derand}

Recall the randomized subroutines of \cite{hknt_local_d1lc} as discussed in \Cref{sec:overview_of_hknt}. Note that all the subroutines succeeds with high probability when the degree of each node is at least $\log ^7 n$. Here, we discuss the derandomization of the subroutines by proving that all of them are $(O(1),\Delta)$-round normal distributed procedures in \Cref{lem:derandomizable-subroutines1} and arguing that such procedures can be efficiently derandomized by using the framework outlined in \Cref{sec:derandomization-framework}.
\begin{lemma}
\label{lem:derandomizable-subroutines1}The pre-shattering part of the algorithm of \cite{hknt_local_d1lc} is a series of $O(\log^* \Delta)$ normal $(O(1), \Delta_G)$-round distributed procedures, such that the final weak success property is that all nodes $v$ with degree $d_{G}(v)\ge \log^7 n$ are properly colored.
\end{lemma}

\begin{proof}
The randomized subroutines used in the pre-shattering part of the algorithm are
{\sc TryRandomColor} (\Cref{alg:try_random_color}),
{\sc GenerateSlack} (\Cref{alg:generate_slack}),
{\sc PutAside} (\Cref{alg:put-aside}), {\sc SynchColorTrial} (\Cref{alg:synch}), and  {\sc SlackColor}. We will show that {\sc TryRandomColor},
{\sc GenerateSlack},
{\sc PutAside}, and {\sc SynchColorTrial}
are all normal $(O(1), \Delta_{G})$-round distributed procedures, and {\sc SlackColor} consists of a series of $O(\log^* \Delta)$ normal $(O(1), \Delta_G)$-round distributed procedures.

First, we note that at the start of any algorithm, nodes may be assumed to have information about which sets they are members of ($V_{\text{sparse}}$, $V_{\text{uneven}}$, etc.), as by earlier simulation arguments ({see \Cref{obs:spse}} ) this can be computed in $O(1)$ rounds in \LOCAL (and \MPC, for low enough degree). We may also assume that nodes have information about the value of any of the parameters in \Cref{def:parameters}  for themselves and any of their neighbours ({see \Cref{lem:computing_parameters}} ). Note that this is, in total, $O(\Delta_{G}^2)$ words of information for each node $v$: a constant number of words for each parameter, $O(\Delta_{G})$ words for set membership, multiplied by the number of neighbors of $v$, which is $O(\Delta_{G})$.

For each of the sub-procedures, we must define strong and weak success properties satisfying the conditions of \Cref{prng-derandomizable}, that capture the notion of success for the subroutine. In each case, our weak success properties will be identical to the corresponding strong success property, other than the extension to deferred nodes. All of our success properties will also deem nodes of degree less then $n^{7\delta}$ to always be successful, regardless of what happens during the randomized process. This means that in this section we will not show any constraints on these low-degree nodes - they will be dealt with afterwards using \Cref{lem:lowdeg}. We look at each procedure:

\vspace{2pt}
\noindent {\sc TryRandomColor} (\Cref{alg:try_random_color}): It takes $O(1)$ rounds of \LOCAL. Nodes need no other words of input information. Nodes only use information from their neighbors, and each node uses $O(\log \Delta_{G})$ random bits to select a color from its palette. Note that the computation is $O(\Delta_{G})$. The output information is either a color with which $v$ has permanently colored itself, or {\sc Fail} if it does not color itself: this is clearly $O(\Delta_{G})$ words of information. We set the success properties $\textsc{SSP}_v$ and $\textsc{WSP}_v$ to be that either the slack of $v$ increases from $c \cdot d_{G}(v)$ for some constant $c$ to $2 \cdot d_{G}(v)$ \cite[Lemma~26]{hknt_local_d1lc}, or $d_{G}(v)< \log^7 n$. Here, for the purposes of $\textsc{WSP}_v$, deferred neighbors are discounted from $v$'s degree and therefore contribute to its slack. So, deferring neighbors only increases $v$'s slack, and therefore the last condition of \Cref{prng-derandomizable} is satisfied. This properties are computable in time linear in the degree of a node and based only on the output of the immediate neighbors of a node. The property succeeds for each node with probability $p=\exp(\Omega(s(v)))$ for nodes with $\log^7 n$ and $1$ otherwise, so this is with high probability in $n$. So, all necessary conditions are satisfied and {\sc TryRandomColor} meets \Cref{prng-derandomizable}.

\vspace{2pt}
\noindent {\sc GenerateSlack} (\Cref{alg:generate_slack}): This takes $O(1)$ rounds of \LOCAL and nodes need no other information at the beginning of the procedure. During the procedure, nodes only use information from their neighbors and $\widetilde{O}(\Delta_{G})$ random bits (to determine whether the node is sampled and if so, what color is attempted). The output is either the color with which $v$ permanently colored itself, or {\sc Fail}. The success properties $\textsc{SSP}_v$ and $\textsc{WSP}_v$ are that either that $v$ generates \emph{sufficient} slack, or $d_{G}(v)< \log^7 n$. By \emph{sufficient} here we mean as described in the appropriate statements \cite[Lemmas~10, 11, 13, 15, 17, 18]{hknt_local_d1lc}. These slack expressions are quite complicated: depending on the type of node, a different guarantee on the eventual slack is required. We note, however, that these guarantees all succeed with high probability for nodes with $d_{G}(v)\ge \log^7 n$ (and again, for lower degree nodes we have defined $\textsc{SSP}_v$ and $\textsc{WSP}_v$ to always be satisfied) and are computable using only information in the immediate neighborhood of $v$. Deferring nodes again creates slack and so only helps nodes, so the success properties satisfy the necessary conditions of \Cref{prng-derandomizable}.

\vspace{2pt}
\noindent    {\sc PutAside} (\Cref{alg:put-aside}): The algorithm consists of $1$ round of \LOCAL, and nodes need no additional information. During the procedure, nodes only use information from their neighbours, and use $O(\polylog(n)) \in O(\Delta_{G})$ (by construction) random bits to independently sample nodes into $S_C$. The output of the procedure is one word (whether $v \in P_C$).    The success properties $\textsc{SSP}_v$ and $\textsc{WSP}_v$ are that either:
     the put-aside set of the almost-clique containing $v$ (say $C$) satisfies $\size{P_C} = \Omega(\ell^2)$ for some polylogarithmic $\ell$ (see \cite[Lemma~5]{hknt_local_d1lc});
    or $C$ contains $\Omega(\ell^2)$ deferred nodes; or
    $d_{G}(v)< \log^7 n$.
     This can be checked using the $2$-hop neighborhood of $v$ because the diameter of $C$ is at most $2$. The strong success property succeeds with high probability, since \cite{hknt_local_d1lc} shows that $\size{P_C} = \Omega(\ell^2)$ with high probability when  $d_{G}(v)\ge \log^7 n$. If we defer nodes then this can reduce the size of the put-aside set by the same amount, but the deferred nodes serve exactly the same purpose of providing slack to the almost clique and so this is still sufficient for the algorithm of \cite{hknt_local_d1lc}. Then, $\textsc{SSP}_v$ and $\textsc{WSP}_v$ must still be satisfied after deferrals, as required.
     
\vspace{2pt}
\noindent    {\sc SynchColorTrial} (\Cref{alg:synch}): This procedure takes $1$ round of \LOCAL. Nodes need no additional information about their neighborhood beforehand. Nodes only use information from their immediate neighbors during the procedure (since leader $x_C$ is adjacent to all nodes in the set of inliers $I_C$), and in the worst case (i.e.~if the node in question is $x_C$), nodes need $\widetilde{O}(\Delta_{G})$ random bits to permute their palette. The output of the procedure for a node $v \in I_C$ for some $C$ is either a color in $\Psi_v$ or {\sc Fail}. The success property for the procedure in \cite{hknt_local_d1lc} is that, for the almost-clique $C$ that $v$ is in, the number of nodes that return {\sc Fail} is bounded by $O(t)$, where $t$ is some polylogarithmic value \cite[Lemma~7]{hknt_local_d1lc}.    So, we accordingly set  $\textsc{SSP}_v$ and $\textsc{WSP}_v$ to be that either:
 $C$ contains $O(t)$ nodes that {\sc Fail}; or
    $d_{G}(v)< \log^7 n$.
    This happens with high probability ($\exp(-t)$ when $d_{G}(v)\ge \log^7 n$, and $1$ otherwise). Deferring nodes cannot increase the number of failed nodes, so $\textsc{SSP}_v$ implies $\textsc{WSP}_v$ after deferrals.
    
\vspace{2pt}
\noindent  {\sc SlackColor} (\Cref{alg:slack_color}):  There are three parts to the algorithm. First, {\sc TryRandomColor} is called in order to amplify the slack of each node (which is linear in the degree: a prerequisite of {\sc SlackColor}). Then two loops of {\sc MultiTrial} (\Cref{alg:multi_trial}) instances are run.    The $O(1)$ calls to {\sc TryRandomColor} are normal $(O(1), \Delta_{G})$-distributed procedures by the item above.
    Each iteration of the first for-loop of executions of {\sc MultiTrial} takes $O(1)$ \LOCAL rounds. Nodes need no additional input information. During the execution of the loop, nodes only use information from their immediate neighbors, and only use $O(\Delta_{G})$ random bits to select a subset of their palette. The output of the procedure for $v$ is either a color in $\Psi_v$ or {\sc Fail}, which is one word. Nodes succeed if $d_i(v) \le \slack{v}/\min(2^{x_i},\rho^\kappa)$, which occurs with high probability when $d_{G}(v)\ge n^{7\delta}$ \cite[Lemma~27]{hknt_local_d1lc}. This is, again, computable using only the output of the node and output information of the immediate neighborhood. Nodes deferring only reduces the degree, and so they cannot cause the property to be unsatisfied. So, setting $\textsc{SSP}_v$ and $\textsc{WSP}_v$ to be that $v$ satisfies the above inequality or has $d_{G}(v)< \log^7 n$ meets the conditions of \Cref{prng-derandomizable}.
        The second for-loop of {\sc MultiTrial} instances are normal distributed procedures for the same reason: the success property is $d_i(v) \le \slack{v}/\min(\rho^{(i+1) \cdot \kappa},\rho)$, and we can set $\textsc{SSP}_v$ and $\textsc{WSP}_v$ analogously to above for the same reason, satisfying the requirements of \Cref{prng-derandomizable}. So, overall \textbf{\sc SlackColor} consists of a sequence of $O(\log^* \Delta)$ normal $(O(1), \Delta_G)$-round distributed procedures.\qedhere

\end{proof}

\color{black}
We have shown that the pre-shattering algorithm of \cite{hknt_local_d1lc} consists of a series of $(O(1), \Delta_{G})$-round distributed procedures as required. However, our success properties do not constrain nodes $v$ with $d_{G}(v)< \log^7 n$, so to reach a full \DILC we must deal with these nodes afterwards. We can do so using the following lemma from \cite{CDP21}:

\begin{lemma}[Lemma 14 of \cite{CDP21}]\label{lem:lowdeg}
For any $n$-node graph $G$ with maximum degree $\Delta = \log^{O(1)} n$, there exists an $O(\log \log \log n)$-round deterministic algorithm for
computing \DILC, using $O(n^\alpha)$ space per machine and $O(n^{1+\alpha})$ global space, for any positive constant~$\alpha \in (0,1)$.
\end{lemma}

Now, we can apply the pre-shattering algorithm of \cite{hknt_local_d1lc}, followed by the low-degree algorithm of \cite{CDP21}. \comments{This} fits into the framework of \Cref{lem:fullderand} and can therefore be derandomized:

\begin{lemma}\label{lem:middegreederand}
There are constants $c,C$ such that, for any constant $\delta>0$,  \Cref{alg:DLDC} performs \DILC deterministically in $O(\log\log\log n)$ rounds of \MPC on any graph $G$ with $\Delta \le n^{7\delta}$, using $\lspace = O(n^{7\delta c C})$ space per machine and global space $O(n_G\cdot n^{7\delta c C})$.
\end{lemma}

\begin{proof}
We have shown that the pre-shattering algorithm of \cite{hknt_local_d1lc} consists of a series of $k=O(\log^* n) = O(\log\log\log n)$ $(c, \Delta_{G})$-round distributed procedures (for some sufficiently large constant $c$), and some deterministic procedures which we can implement efficiently in \MPC (see \Cref{sec:randomized-d1lc}). We follow this with the $O(\log\log\log n)$-round deterministic low-space MPC procedure of \Cref{lem:lowdeg}. Setting $\alpha$ (in \Cref{lem:lowdeg}) sufficiently lower
than $\delta$, by \Cref{lem:fullderand}, this algorithm can be derandomized in $O(\log\log\log n)$ rounds of \MPC using $O(n^{7\delta c C})$ space per machine and $O(n_G\cdot n^{7\delta c C})$  global space.
\end{proof}

\begin{algorithm}[H]
\caption{\textsc{DerandomizedMidDegreeColor}$(G)$}
\label{alg:DLDC}
	
 Let $\mathcal A$ consist of the pre-shattering randomized \LOCAL of \cite{hknt_local_d1lc} followed by the $O(\log\log\log n)$-round deterministic low-space \MPC algorithm of \cite{CDP21}.
		
 Derandomize $\mathcal A$, using the success properties from the proof of \Cref{lem:derandomizable-subroutines1}, by \Cref{lem:fullderand}.
\end{algorithm}

\hide{
\subsection{Deterministic algorithm when degree is at most $n^{7\delta}$}
In this section, in the following lemma, we first argue that there is a deterministic $O(\log^* n)$ \MPC algorithm for \DILC on a graph $G$ of maximum degree at most $O(n^{7\delta})$ such that the algorithm colors a subset of nodes and the graph induced by the uncolored nodes has maximum degree is at most $\log^7 n$.
\begin{lemma}\label{lem:phasecolor}
There exists an $O(\log^* n)$ round deterministic \MPC algorithm for for \DILC running on a graph $G$, of maximum degree \emph{at most} $n^{7\delta}$, such that it colors a subset of nodes and the subgraph induced by the uncolored nodes has maximum degree is $\log^7 n$. The process uses $O(n^{7\delta cC})$ space per machine and $O(n_G \cdot n^{7\delta cC})$ global space, where $C$ is as in \Cref{lem:phasederand}.
\end{lemma}

\begin{proof}
We use $\Delta = n^{7\delta}$ as an upper bound on the maximum degree $G$, and first compute the $O(\Delta^{8c})$-coloring of $G^{4c}$ required for \Cref{lem:fullderand}. This can be done in $O(\log^* n)$ rounds by simulating Linial's \cite{Linial92} deterministic \LOCAL algorithm round-by-round. Then, since we showed in \Cref{lem:derandomizable-subroutines} that the algorithm of \cite{hknt_local_d1lc} can be expressed as a sequence of $O(\log^* n)$ \emph{phases}, each of which are normal $(c, \Delta)$-round distributed procedures for some constant $c$ when the degree of the nodes is at least $\log ^7 n$, we can apply \Cref{lem:fullderand} to derandomize each phase in turn, resulting in a deterministic coloring of a subset of nodes.

Since the note that our repeated derandomization steps can leave a set $D$ of remaining uncolored nodes with degrees of at most $\log^7 n$.
\end{proof}


Now the desired algorithm is simple: we apply \Cref{lem:phasecolor} to deterministically color all nodes of degree at least $\log^7 n$ (with some nodes potentially remaining uncolored once their uncolored degree drops below $\log^7 n$) and then apply the following lemma to color the remaining low-degree nodes:

\begin{lemma}[Lemma 14 of \cite{CDP21}]\label{lem:lowdeg}
For any $n$-node graph $G$ with maximum degree $\Delta = \log^{O(1)} n$, there exists an $O(\log \log \log n)$-round deterministic algorithm for
computing \DILC, using $O(n^\phi)$ space per machine and $O(n^{1+\phi})$ global space, for any positive constant~$\phi \in (0,1)$.
\end{lemma}

The algorithm for \DILC when the degree is $O(n^{7\delta})$ is described in \Cref{alg:DLDC} and its correctness is proved in \Cref{thm:midcolor}.

\begin{theorem}\label{thm:midcolor}
\textsc{DerandomizedMidDegreeColor}, performed on a DL1C instance $G$ of degree at most $n^{7\delta}$, properly colors all nodes in $O(\log \log\log n)$ rounds of \MPC, using $O(n^{7\delta cC})$ space per machine and $O(n_{G} \cdot n^{7\delta cC})$ global space, where $C$ is as in \Cref{lem:phasederand}.
\end{theorem}

\begin{proof}
After applying \Cref{lem:phasecolor} the degree of all uncolored nodes in $G$ must also be at most $\log^7 n$, so all remaining nodes are colored by \Cref{lem:lowdeg}.

The space usage is dominated by that of \Cref{lem:phasecolor}, and the number of rounds required is $O(\log\log\log n)$ since it is dominated by the $O(\log\log\log n)$ rounds for \Cref{lem:lowdeg}.
\end{proof}
}



\subsection{Simulating deterministic \LOCAL Subroutines of \cite{hknt_local_d1lc}}
\label{sec:randomized-d1lc}

Here, we provide some arguments that the subroutines in \cite{hknt_local_d1lc} which are already deterministic can be simulated efficiently in the \MPC model, provided that the maximum degree of the \DILC instance is below $\sqrt{\lspace}$.
In particular we show that:

\begin{lemma}
\label{lem:deterministic_subroutines_in_mpc}
    In the sublinear local space \MPC model with local space $\lspace = O(n^\phi) \geq \Delta^2$ and global space $O(m+n^{1+\phi})$, there exist algorithms running deterministically in $O(1)$ rounds computing each of the following:
    \begin{itemize}
        \item all the parameters (see \Cref{def:parameters}),
        \item a $(\deg+1)$-ACD (see \Cref{def:acd}),
        \item the set $V_{\text{start}}\subseteq V_{\text{sparse}}$~\footnote{A full definition of $V_{\text{start}}$ will be given later in this section.} (whose purpose was described in Section~\ref{sec:sparse}),
        \item the leaders, inliers, and outliers from each almost-clique,
        \item and the put-aside sets for each almost-clique.
    \end{itemize}
\end{lemma}

The rest of the section is concerned with the proof of each of the components of \Cref{lem:deterministic_subroutines_in_mpc}.

We begin with a useful tool that says (in essence) that nodes can send and receive $O(\sqrt{\lspace})$ messages to each of their neighbours in one round\footnote{We remark that the \MPC model is not node-centric, but machine-centric, and that when we talk about node $u$ ``sending messages'' to node $v$, we suppose that each node has a machine responsible for it, and mean that the machine responsible for node $u$ is sending a message to the machine responsible for node $v$.}, provided that their degree is at most $\sqrt{{\lspace}}$.

\begin{lemma}
\label{lem:mpc_subroutines_low_degree}
Let $G=(V, E)$ be a graph $G$ such that for all $v \in V$, $d(v) \leq \sqrt{\lspace}$.

Suppose each node $v$ is assigned some unique machine $M(v)$: informally, $M(v)$ is responsible for $v$. Then, the following subroutines can be performed in parallel for each $v$ (and its assigned machine $M(v)$) in $O(1)$ rounds in \MPC with $\lspace = O(n^\spacexp)$, and global space $O(m+n^{1+\spacexp})$:

\begin{itemize}
    \item $M(v)$ can send $d(v)$ messages to each machine $M'$ such that $M(u)$ is responsible for node $u \in N(v)$
    \item $M$ can collect all of the edges between neighbours of $v$
\end{itemize}
\end{lemma}
\begin{proof}
Suppose each node is allocated a machine (giving global space of $O(m+n^{1+\phi})$. The machine responsible for node $v$ can prepare $d(v)$ messages for each of the neighbors of $v$. This takes $\leq \lspace$ local space, so these messages can be sent in $O(1)$ rounds. The second subroutine follows from the first: the $d(v)$ words that each node broadcasts consists of the list of its neighbors.
\end{proof}

We can now argue that provided the maximum degree of our graph is low enough, the parameters described in \Cref{def:parameters} can be efficiently computed in parallel for each node.

\begin{lemma}[Computing Parameters of \Cref{def:parameters}]
\label{lem:computing_parameters}
  For a graph $G=(V, E)$ with maximum degree {${\Delta \leq O(\sqrt{\lspace})}$}, all of the parameters listed in \Cref{def:parameters} can be computed for every node in parallel in \MPC with $\lspace = O(n^\spacexp)$ and global space $O(m+n^{1+\spacexp})$ in $O(1)$ rounds.
\end{lemma}
\begin{proof}
We extensively use the subroutines in \Cref{lem:mpc_subroutines_low_degree}. We deal with the parameters in the order they appear in \Cref{def:parameters}.

\begin{itemize}
\item \textbf{Slack:} In $O(1)$ rounds, we can sort tuples corresponding to edges incident to $v$ and colors in $\Psi(v)$ such that they appear on consecutive machines: we can then count how many of each there are and compute \slack{v}. 
\item \textbf{Sparsity:} Since we can  compute $d(v)$ for all nodes in constant rounds (machines can collect all edges incident to $v$ and count them), to compute $\sparsity{v}$ it remains to count the edges between neighbors of $v$. By the second subroutine in \Cref{lem:mpc_subroutines_low_degree} we can collect the $2$-hop neighborhood of $v$ on a single machine. We can then remove duplicates if necessary and count the number of edges, and then calculate \sparsity{v}.
\item \textbf{Disparity:} By the first subroutine of \Cref{lem:mpc_subroutines_low_degree} we can collect the palette of each of the neighbors of $v$ in constant rounds, allowing us to calculate \disparity{u}{v} for each neighbor $u$ of~$v$.
\item \textbf{Discrepancy:} Follows immediately from the ability to collect \disparity{u}{v} for all neighbors $u$ of~$v$.
\item \textbf{Unevenness:} It suffices to collect $d(u)$ from each neighbor $u$ of $v$, which can be done using the first subroutine of \Cref{lem:mpc_subroutines_low_degree}.
\item \textbf{Slackability / Strong Slackability:} Note that these are additions of already computed parameters, and hence can be computed in constant rounds.
\end{itemize}
Mainly due to Lemma~\ref{lem:mpc_subroutines_low_degree} and from the above description, it is easy to see that each of the above parameters can be computed in $O(1)$ rounds for all the nodes in parallel. For the global space complexity for computing the above parameters, note that each node $v$ collects it $2$-neighborhood which is of size $O(\lspace)$ as maximum degree of any node is $O(\sqrt{\lspace})$.
\end{proof}

Next, we show how to compute an $(\deg +1)$-ACD in $O(1)$ rounds.

\begin{lemma}[Computing $(\deg+1)$-ACD]
\label{lem:computing_acd}
For a graph $G=(V, E)$ with maximum degree $\Delta \leq \sqrt{\lspace}$, an $(\deg+1)$-ACD can be computed in \MPC with $\lspace = O(n^\spacexp)$ and global space $O(m+n^{1+\phi})$ in $O(1)$ rounds.
\end{lemma}
\begin{proof}
First we argue that the diameter of the subgraph induced by any almost-clique $C$, is at most $2$. This is the case because, by (iv) of \Cref{def:acd}, the number of neighbors any node $v \in C$ has in $C$ is $\size{N(v) \cap C} \geq \frac{\size{C}} {1+ \veps_{ac}} > \frac{\size{C}}{2}$. Hence, any two nodes $u$ and $v$ in $C$ have a common neighbor in $C$.

By \Cref{lem:computing_parameters}, all nodes ($v$) can compute their sparsity ($\zeta_v$) and unevenness ($\eta_v$) in $O(1)$ rounds in parallel. From the values of $\zeta_v$ and $\eta_v$, we can decide whether $v \in V_{\text{sparse}}$, or $v \in V_{\text{uneven}}$, or neither, giving the following:
\begin{observation}\label{obs:spse}
{$V_{\text{sparse}}$ and $V_{\text{uneven}}$ can be determined in $O(1)$ rounds.}
\end{observation}

Note that each node $v \in (V \setminus (V_{\text{sparse}} \sqcup V_{\text{uneven}}))$ is present in some almost-clique $C(v)$. As the diameter of $G[C(v)]$ is at most 2, $C(v)$ is a subset of the $2$-hop neighborhood of $v$. As $d(v)\leq \sqrt{\lspace}$, by \Cref{lem:mpc_subroutines_low_degree} the $2$-hop neighborhood of $v$ can be found in $O(1)$ rounds and stored on a single machine. Once this is done, $M(v)$ can determine which almost-clique $v$ is in. Putting everything together, the ACD can be computed in $O(1)$ rounds and the global space used is $O(m+n^{1+\phi})$.
\end{proof}


We now consider the identification of $V_{\text{start}}$, the subset of $V_\text{sparse} \cup V_\text{uneven}$ for which it is hard to generate slack. To identify this set, we consider a breakdown of $V_\text{sparse} \cup V_\text{uneven}$ into several sets based on the parameters defined in \Cref{def:parameters}.
We first explain the notion of a ``heavy'' color. A heavy color $c$ with respect to a node $v$ is a color such that if all of the neighbors of $v$ were to pick a color from their palettes uniformly at random, the expected number of neighbors of $v$ which would pick $c$ (which we denote $H(c)$) is at least some suitable constant. We denote by $\mathcal{C}_v^{\text{heavy}}$ the set of heavy colors with respect to $v$. Now $V_{\text{start}}$ is defined as follows (taken from \cite{hknt_local_d1lc}; the $\veps_i$ are all constants):


\begin{align*}
    V_{\text{balanced}} &= \{v \in V_{\text{sparse}} : |  \{u \in N(v) : d(u) > 2d(v)/3| \geq \veps_1 d(v)\}\}. \\
    V_{\text{disc}} &= \{v \in V_{\text{sparse}} : \discrepancy{v} \geq \veps_2 d(v)\}. \\
    V_{\text{easy}} &=
    \begin{aligned}[t]
        & V_{\text{balanced}} \cup V_{\text{disc}} \cup V_{\text{uneven}}\\
        & \cup \{v \in V_{\text{sparse}} : | N(v) \cap V_{\text{dense}} \geq \veps_3 d(v) \}.
    \end{aligned}\\
    V_{\text{heavy}} &= \{v \in V_{\text{sparse}} \setminus V_{\text{easy}} : \textstyle \sum_{c \in \mathcal{C}_v^{\text{heavy}}} H(c) \geq \veps_4 d(v)\}. \\
    V_{\text{start}} &= \{v \in V_{\text{sparse}} \setminus (V_{\text{easy}} \cup V_{\text{heavy}}) : N(v) \cap V_{\text{easy}} | \geq \veps_5 d(v)\}.
\end{align*}

\begin{lemma}[Identifying $V_{\text{start}}$]
\label{lem:identifying_vstart}
Given a graph $G=(V, E)$ with maximum maximum degree $\Delta \leq \sqrt{\lspace}$, the set $V_{\text{start}}$ can be identified on an \MPC with $\lspace = O(n^\spacexp)$ and global space $O(m+n^{1 + \phi})$ in $O(1)$ rounds.
\end{lemma}
\begin{proof}
Observe that $V_{\text{balanced}}$, $V_{\text{disc}}$, and $V_{\text{easy}}$ are computable in $O(1)$ rounds by applications of \Cref{lem:computing_parameters} and \Cref{lem:mpc_subroutines_low_degree}. We briefly note that, again using an application of the first subroutine of \Cref{lem:mpc_subroutines_low_degree}, nodes can be made aware of which neighbors are in which sets in $O(1)$ rounds.

It remains to compute which colors are heavy (and the set $V_{\text{heavy}}$): by the first subroutine of \Cref{lem:mpc_subroutines_low_degree} we can gather the palettes of all neighbors of $v$ on $M(v)$, and then $M(v)$ can compute this information in $O(1)$ rounds.
\end{proof}

In the following lemma, we show that we can identify leaders, inliers, and outliers of almost-cliques, and establish which almost-cliques require put-aside sets, in $O(1)$ rounds.

\begin{lemma}[Finding leaders, outliers and inliers]
\label{lem:finding_leaders}
Given a graph $G=(V, E)$ with maximum maximum degree $\Delta \leq \sqrt{\lspace}$, we can find the leader, the set of inliers and the set of outliers for each almost-clique in \MPC with $\lspace = O(n^\spacexp)$ and global space $O(m+n^{1 + \phi})$ in $O(1)$ rounds. Moreover, we can detect all the almost-cliques for which we need to find put-aside sets.
\end{lemma}
\begin{proof}

First, we explain how we can find the \emph{leader} $x_{C}$ and the \emph{outliers} $O_{C}$ for each almost-clique $C \in \cC$. The leader $x_{C}$ is the node in $C$ with minimum slackability. As we can determine the slackability of all the nodes (in $C$) by using $O(1)$ rounds in parallel (see \Cref{lem:computing_parameters}), we can determine $x_{C}$ for each $i \in [t]$ in parallel in $O(1)$ rounds. The slackability of an almost-clique $C$ is defined as the the slackability of its leader $x_C$. $C$ is said to be a \emph{low slack} almost-clique if its slackability is at most $\ell=\log^{2.1} \Delta$. Recall that these are the almost-cliques (with low slackability) for which we need to find put-aside sets.

Let $O_C$ be the set of \emph{outliers} in an almost-clique $C$. $O_C$ is comprised of the union of the $\frac{\max\{{d(x_C), \size{C}}\}}{3}$ nodes with the fewest common neighbors with $x_C$, the $\frac{\size{C}}{6}$ nodes of largest degree, and the nodes in $C$ that are not neighbors of $x_C$. The nodes in $C$ that are not outliers are called \emph{inliers} of $C$ (denoted by $I_C$). As each almost-clique $C$ is stored in one machine and we can determine the degrees of the nodes in $O(1)$ rounds, we can clearly find $O_C$ and $I_C$ for all almost-cliques in $O(1)$ rounds. 
\end{proof}

Finally, we conclude by observing that we have shown each component of \Cref{lem:deterministic_subroutines_in_mpc}:

\begin{proof}[Proof~of~\Cref{lem:deterministic_subroutines_in_mpc}]
     Follows from \Cref{lem:computing_parameters,lem:computing_acd,obs:spse,lem:identifying_vstart,lem:finding_leaders}.
\end{proof}

\section{Overall Deterministic D1LC Algorithm}
\label{sec:d1lc_for_all_degree_ranges}

In this section we extend our (now derandomized) algorithm for \DILC (as discussed in \Cref{alg:DLDC}) to handle instances with maximum degree $\Delta \geq n^{7\delta}$.
The idea is to reduce our \DILC instance to a collection of instances with lower degree so that constant-radius neighborhoods of any node fit onto a single machine.
Instead of guaranteeing that there are not too many instances, we ensure that the sequential dependency between instances is not too long. That is, many of the instances resulting from our decomposition can be colored in parallel; there are only $O(1)$ sets of base-case instances which \comments{must} be solved sequentially.

We use a recursive structure \textsc{LowSpaceColorReduce} (\Cref{alg:LSColorReduce}) similar to \cite{CDP20,CDP21}. The recursive structure in \Cref{alg:LSColorReduce} relies on a partitioning procedure (\Cref{alg:LSPartition})  to divide the nodes and colors in the input instance into \emph{bins}. 
We can follow the analysis of \cite{CDP20} to analyze the partitioning process, since it is the base case that changed. \Cref{lem:LSdeterministic-hashing} provides the important properties of the partitioning.

\begin{algorithm}
	\caption{\textsc{LowSpaceColorReduce}$(G)$}
	\label{alg:LSColorReduce}
		$G_{\text{mid}},G_1 \dots, G_{n^{\delta}} \gets \textsc{LowSpacePartition}(G)$.
		
	    For each $i = 1, \dots,n^{\delta}-1$ in parallel: call \textsc{LowSpaceColorReduce}$(G_i)$.
	
		Update color palettes of $G_{n^{\delta}}$, call \textsc{LowSpaceColorReduce}$(G_{n^{\delta}})$.
		
		Update color palettes of $G_{\text{mid}}$.
		
		Color $G_{\text{mid}}$ using \textsc{DerandomizedMidDegreeColor}$(G_{\text{mid}})$.
\end{algorithm}

\begin{algorithm}
\caption{\textsc{LowSpacePartition}$(G)$}
\label{alg:LSPartition}
Let $G_{\text{mid}}$ be the graph induced by the set of nodes $v$ with $d(v)\le n^{7\delta}$.
		
Let hash function $h_1:[n]\rightarrow [n^{\delta}]$ map each node $v$ to a bin $h_1(v) \in [n^{\delta}]$.
		
Let hash function $h_2:[n^2]\rightarrow [n^{\delta}-1]$ map colors $\gamma$ to a bin $h_2(\gamma) \in [n^{\delta}-1]$.
		
Let $G_1,\dots,G_{n^{\delta}}$ be the graphs induced by bins $1,\dots,n^{\delta}$ respectively, minus the nodes in $G_{\text{mid}}$.
		
Restrict palettes of nodes in $G_1,\dots,G_{n^{\delta}-1}$ to colors assigned by $h_2$ to corresponding bins.
    	
Return $G_{\text{mid}},G_1,\dots,G_{n^{\delta}}$.
\end{algorithm}

\begin{lemma}[Lemma 4.6 of \cite{CDP20}]\label{lem:LSdeterministic-hashing}
Assume that, at the beginning of a call to \textsc{LowSpacePartition}, we have $d(v)<p(v)$ for all nodes $v$. Then, in $O(1)$ \MPC rounds with $O(n^{7\delta})$ local space per machine and $O(n+m)$ global space (over all parallel instances), one can deterministically select hash functions $h_1$, $h_2$ such that after the call,
\begin{itemize}
    \item for any node $v\notin G_{\text{mid}}$, $d'(v) < 2d(v)n^{-\delta}$, and
		\item for any node $v$, $d'(v)<p'(v)$.
\end{itemize}
   
 Here $d'(v)$ denotes the degree of $v$ in the subgraph induced by the nodes present in the same bucket as $v$ and $p'(v)$ denotes the number of $v$'s palette colors that are in the same bucket as~$v$.
\end{lemma}

Now, we are ready to prove our main result (\Cref{thm:main:deter}):
\begin{proof}[Proof of \Cref{thm:main:deter}]
As in \cite{CDP21}, calling \textsc{LowSpaceColorReduce} on our input graph creates a recursion tree of $O(1)$ depth (since each recursive call reduces the maximum degree by a $n^{-\delta}$ factor). It therefore creates $O(1)$ sequential sets of base-case instances to solve concurrently, which in \Cref{alg:LSColorReduce} are solved by \textsc{DerandomizedMidDegreeColor}. Furthermore, each set of concurrent instances has at most $n$ nodes in total, since all nodes are only partitioned into one instance, and each instance has maximum degree $n^{7\delta}$.

By \Cref{lem:middegreederand}, each such instance $G$ is colored in $O(\log\log\log n)$ rounds using $O(n^{7\delta cC})$ space per machine and $O(n_G\cdot n^{7\delta cC})$ global space. The global space used by all concurrent instances is therefore $O(n^{1+7\delta cC})$. Setting $\delta\le \frac{\spacexp}{7cC}$, this is $O(n^\spacexp)$ space per machine and $O(n^{1+\spacexp})$ global space. Since receiving the input and the first call to \textsc{LowSpacePartition} also requires $O(m)$ global space, the overall space bound is $O(m+n^{1+\spacexp})$.
\end{proof}

\bibliographystyle{alpha}
\bibliography{ref}

\end{document}